%% file: main.tex
\documentclass[journal]{IEEEtran}

\newif\ifIFAC
\IFACfalse

\input{pre_rob.tex}

\usepackage{nicefrac}

\newcommand{\asref}[1]{Assumption~\ref{#1}}

\usepackage{balance}
\usepackage{cite}

\usepackage{algorithm,algorithmic}

\newcommand{\rpolicy} {G} 
\newcommand{\Bysum} {L} 

\newcommand{\fact} {u} 
\newcommand{\factset} {\mathcal{U}} 
\newcommand{\factdim} {U} 

\newcommand{\vspan}{\operatorname{span}\PrfencesSoft}

\title{\huge Inverse Filtering for Hidden Markov Models with \\ Applications to Counter-Adversarial Autonomous Systems}

\begin{document}

\author{Robert~Mattila,~\IEEEmembership{Student Member,~IEEE,}
        Cristian~R.~Rojas,~\IEEEmembership{Member,~IEEE,}\\
        Vikram~Krishnamurthy,~\IEEEmembership{Fellow,~IEEE,}
        and~Bo~Wahlberg,~\IEEEmembership{Fellow,~IEEE}
        \thanks{This work was supported by the Swedish Research Council (2016-06079) and
the U.S. Air Force Office of Scientific Research (FA9550-18-1-0007).}
\thanks{R.~Mattila, C.~R.~Rojas and B.~Wahlberg are with the Division of Decision and
    Control Systems, School of Electrical Engineering and Computer Science, KTH Royal
    Institute of Technology, Stockholm, Sweden. E-mails: {\tt\footnotesize \{rmattila,
        crro, bo\}@kth.se}.}
\thanks{V.~Krishnamurthy is with the Department of Electrical
    and Computer Engineering, Cornell University, Ithaca, New York, USA. E-mail: {\tt\footnotesize
    vikramk@cornell.edu}.}
}


\maketitle

\begin{abstract}

    Bayesian filtering deals with computing the posterior distribution of the state of
    a stochastic dynamic system given noisy observations. In this paper, motivated by
    applications in counter-adversarial systems, we consider the following inverse
    filtering problem: Given a sequence of posterior distributions from a Bayesian filter,
    what can be inferred about the transition kernel of the state, the observation
    likelihoods of the sensor and the measured observations? For finite-state Markov
    chains observed in noise (hidden Markov models), we show that a least-squares fit for
    estimating the parameters and observations amounts to a combinatorial optimization
    problem with non-convex objective. Instead, by exploiting the algebraic structure of
    the corresponding Bayesian filter, we propose an algorithm based on convex
    optimization for reconstructing the transition kernel, the observation likelihoods and
    the observations. We discuss and derive conditions for identifiability. As an
    application of our results, we illustrate the design of counter-adversarial systems:
    By observing the actions of an autonomous enemy, we estimate the accuracy of its
    sensors and the observations it has received. The proposed algorithms are evaluated in
    numerical examples.

\end{abstract}

\begin{IEEEkeywords}%
    inverse filtering, hidden Markov models, counter-adversarial autonomous systems,
    remote calibration, adversarial signal processing
\end{IEEEkeywords}


\section{Introduction}

\IEEEPARstart{I}{n} a partially observed stochastic dynamic system, the state is hidden in
the sense that it can only be observed in noise via a sensor. Formally, with $p$
denoting a probability density (or mass) function, such a system is represented by the
conditional densities: 
\begin{align}
    x_k &\sim P_{x_{k-1}, x} = p(x | x_{k-1}), \quad x_0 \sim \pi_0, \label{eq:hmm_state} \\
    y_k &\sim B_{x_k, y} = p(y | x_k), \label{eq:hmm_observation}
\end{align}
where by $\sim$ we mean ``distributed according to'' and $k$ denotes discrete time. In
\eref{eq:hmm_state}, the state $x_k$ evolves according to a Markovian transition kernel
$P$ on state-space $\mathcal{X}$, and $\pi_0$ is its initial distribution. In
\eref{eq:hmm_observation}, an observation $y_k$ (in observation-space $\mathcal{Y}$) of
the state is measured at each time instant according to observation likelihoods $B$. An
important example of \eref{eq:hmm_state}-\eref{eq:hmm_observation}, where
\eref{eq:hmm_state} is a finite-state Markov chain, is the so called \emph{hidden Markov
model} (HMM) \cite{krishnamurthy_partially_2016, cappe_inference_2005}.

In the Bayesian (stochastic) filtering problem \cite{anderson_optimal_1979}, one seeks to
compute the conditional expectation of the state given noisy observations by evaluating
a recursive expression for the posterior distribution of the state:
\begin{equation}
    \pi_k(x) = p(x_k = x | y_1, \dots, y_k), \qquad x \in \mathcal{X}.
    \label{eq:state_posterior}
\end{equation}
The recursion for the posterior is given by  the Bayesian filter
\begin{equation}
    \pi_k = T(\pi_{k-1}, y_k; P, B),
    \label{eq:bayesian_filter_recursion}
\end{equation}
where
\begin{equation}
    \{T(\pi, y ; P, B)\}(x) = \frac{B_{x, y} \int_\mathcal{X} P_{\zeta, x} \pi(\zeta)
    d\zeta}{\int_\mathcal{X} B_{x, y} \int_\mathcal{X} P_{\zeta, x} \pi(\zeta)
d\zeta dx}, \quad x \in \mathcal{X},
    \label{eq:bayesian_filter}
\end{equation}
-- see, e.g., \cite{cappe_inference_2005, krishnamurthy_partially_2016} for derivations
and details. Two well known finite dimensional cases of \eref{eq:bayesian_filter} are the
Kalman filter, where the dynamical system \eref{eq:hmm_state}-\eref{eq:hmm_observation} is
a linear Gaussian state-space model, and the HMM filter, where the state is a finite-state
Markov chain.  

In this paper, we treat and provide solutions to the following inverse filtering problem:
\begin{quote}

    \emph{Given a sequence of posteriors $\pi_1, \dots, \pi_N$ from the filter
    \eref{eq:bayesian_filter_recursion}, reconstruct (estimate) the filter's parameters:
the system's transition kernel $P$, the sensor's observation likelihoods $B$ and the
measured observations $y_1, \dots, y_N$.}

\end{quote}

An important motivating application is the design of counter-adversarial
systems\cite{kuptel_counter_2017, krishnamurthy_how_2019, mattila_what_2020}: Given
measurements of the actions of a sophisticated autonomous adversary, how to remotely
calibrate (i.e., estimate) its sensors and predict, so as to guard against, its future
actions? We refer the reader to \fref{fig:full_caa_schematic_simple} on the next page for
a schematic overview.

This paper extends our recent work \cite{mattila_inverse_2017, krishnamurthy_how_2019,
mattila_what_2020} in two important ways: First, \cite{mattila_inverse_2017,
krishnamurthy_how_2019, mattila_what_2020} assumed that both the enemy and us know the
transition kernel $P$. In reality, if we generate the signal $x_k$, then the enemy
estimates $P$ (e.g., maximum likelihood estimate) and we have to \emph{estimate the
enemy's estimate} of $P$. The first part of this paper constructs algorithms for doing
this based on observing (intercepting) posterior distributions.  In the second part, we
consider the generalized setting where the enemy's posteriors are observed in noise via
some policy. Second, \cite{mattila_inverse_2017, krishnamurthy_how_2019,
mattila_what_2020} did not deal with identifiability issues in inverse filtering. The
current paper gives necessary and sufficient conditions for identifiability of $P$ and $B$
given a sequence of posteriors.

\begin{figure*}[ht!]

    \hspace{0.4cm} \includegraphics{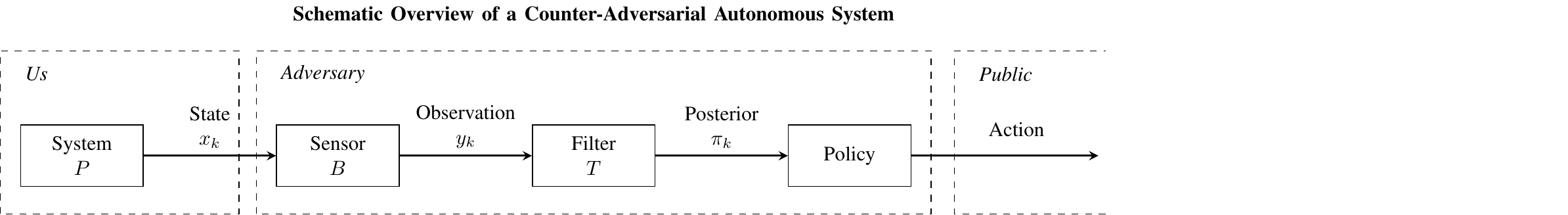}

\caption{An autonomous adversary measures our state $x_k$ as $y_k$ via a sensor $B$.
    A stochastic filter $T$ is used to compute its posterior $\pi_k$ of our state. Based
    on this posterior and a policy, a public action is taken. A counter-adversarial
    autonomous system aims to devise a system and associated algorithms that can infer
    information private to the adversary. This information forms a basis for predicting
    and taking suitable measures against future actions. In this paper, we consider the
    remote sensor calibration problem (\pref{pr:remote_calibration}), where the goal is to
estimate the adversary's sensor $B$ based on observed actions.}

    \label{fig:full_caa_schematic_simple}

\end{figure*}

In addition, the filter \eref{eq:bayesian_filter} is a crucial component of many
engineering systems; success stories include, for example, early applications in aerospace
\cite{mcgee_discovery_1985} and, more recently, the \emph{global positioning system} (GPS;
\cite{kaplan_understanding_2005}). It can be difficult, or even impossible, to access raw
sensor data in integrated smart sensors since they are often tightly encapsulated.  The
ability to reverse engineer the parameters of a filtering system from only its output
suggests novel ways of performing fault detection and diagnosis (see, e.g.,
\cite{sundvall_fault_2006, wahlberg_observers_2009} for motivating examples) -- the most
obvious being to compare reconstructed parameters to their nominal values, or perform
change-detection on multiple data batches. Moreover, in cyber-physical security
\cite{li_jamming_2015}, the algorithms proposed in this paper could be used to detect
malicious attacks by an intruder of a control system.

\subsection{Main Results and Outline}

To construct a tractable analysis, we consider the case where
\eref{eq:hmm_state}-\eref{eq:hmm_observation} constitute a \emph{hidden Markov model}
(HMM) on a finite observation-alphabet. The main results of this paper are:
\begin{itemize}

    \item We analyze the uniqueness of the updates of the stochastic filter
        \eref{eq:bayesian_filter} for HMMs (Theorems~\ref{thrm:Bayesian_filter_unique} and
        \ref{thrm:Bayesian_filter_unique_subset}), and derive an alternative
        characterization (\thrmref{thrm:hmm_filter_nullspace_PB}) that highlights
        important structural properties.

    \item We introduce the nullspace clustering problem (\pref{pr:nullspace_clustering})
        -- which is complementary to the subspace clustering problem
        \cite{vidal_subspace_2011} -- and propose an algorithm based on the group LASSO
        \cite{yuan_model_2006} to solve it. In \thrmref{thrm:reconstruct_P_and_B}, we
        detail a procedure to uniquely factorize unnormalized nullspaces into HMM parameters.

    \item By leveraging the previous two points we demonstrate how the transition kernel
        as well as the observation likelihoods of an HMM can be reconstructed from
        a sequence of posteriors (\algoref{alg:full_inverse_filtering}); then, the
        corresponding sequence of observations can trivially be reconstructed
        (Remark~\ref{rem:obtaining_observations}).

    \item We apply our results to the remote sensor calibration problem
        (\pref{pr:remote_calibration}) for counter-adversarial autonomous systems. Even in
        a mismatched setting (i.e., where the adversary employs uncertain estimates
        $\hat{P}$ and $\hat{B}$ in the filter updates), we can estimate the adversary's
        sensor, and that too regardless of the quality of its estimates.

    \item Finally, the performance of our proposed inverse filtering algorithms is
        demonstrated in numerical examples, where we find that a surprisingly small
        amount of posteriors is sufficient to reconstruct the sought parameters. 


\end{itemize}

The paper is structured as follows. \sref{sec:preliminaries} formulates the problems we
consider, discusses identifiability, and shows that a direct approach is computationally
infeasible for large data sizes. Our proposed inverse filtering algorithms are given in
\sref{sec:proposed method}. In \sref{sec:application_caa}, we consider the design of
counter-adversarial autonomous systems and show how an adversary's sensors can be
estimated from its actions.  The proposed algorithms are evaluated in \sref{sec:numerical}
in numerical examples.  Detailed proofs and algebraic manipulations are available in the
supplementary material.

\subsection{Related Work}

Kalman's inverse optimal control paper \cite{kalman_when_1964} from 1964, aiming to
determine for what cost criteria a given control policy is optimal, is an early example of
an inverse problem in signal processing and automatic control. More recently, an interest
for similar problems has been sparked in the machine learning community with the success
of topics such as inverse reinforcement learning, imitation learning and apprenticeship
learning \cite{hadfield-menell_cooperative_2016, choi_nonparametric_2012,
klein_inverse_2012, levine_nonlinear_2011, ng_algorithms_2000} in which an agent learns
by observing an expert performing a task.

Variations of inverse filtering problems can be found in the microeconomics literature
(social learning; \cite{chamley_rational_2004}) and the fault detection literature (e.g.,
\cite{gertler_fault_1998, gustafsson_adaptive_2000, gustafsson_statistical_2007,
chen_robust_1999}), where the stochastic filter is a standard tool. For example, the
works \cite{sundvall_fault_2006, wahlberg_observers_2009} were motivated by fault
detection in mobile robots and aimed to reconstruct sensor data from from state estimates
by constructing an extended observer.

To the best of the authors' knowledge, the specific inverse filtering problem we consider
-- reconstructing system and sensor parameters directly from posteriors -- was first
introduced in \cite{mattila_inverse_2017} for HMMs, and later discussed for linear
Gaussian state-space models in \cite{mattila_inverse_2018}. In both these papers, strong
simplifying assumptions were made. In contrast to the present work, it was assumed that
\emph{i)} the transition kernel $P$ of the system was known, and that \emph{ii)} the
system and the filter were matched in the sense that the update $T(\pi_{k-1}, y_k; P, B)$
was used and not the more realistic mismatched $T(\pi_{k-1}, y_k; \hat{P}, \hat{B})$,
where $\hat{P}$ and $\hat{B}$ denote estimates. The algorithms we propose in this paper
extend \cite{mattila_inverse_2017} to not require knowledge of the transition dynamics,
and are agnostic to whether the filter is mismatched or not.

The latter is of crucial importance when applying inverse filtering algorithms in
counter-adversarial scenarios \cite{kuptel_counter_2017, krishnamurthy_how_2019,
mattila_what_2020}. In such, an adversary is trying to estimate our state (via Bayesian
filtering) and does not, in general, have access to our transition kernel -- recall the
setup from \fref{fig:full_caa_schematic_simple}. Hence, its filtering system is mismatched
(e.g., a maximum likelihood estimate $\hat{P}$ computed by the adversary is used instead
of the true $P$). Compared to \cite{krishnamurthy_how_2019, mattila_what_2020} that aim to
estimate information private to the adversary, the present work does not assume knowledge
of the adversary's filter parameters, nor that its filtering system is matched.

\section{Preliminaries and Problem Formulation}
\label{sec:preliminaries}

In this section, we first detail our notation and provide necessary background material on
hidden Markov models and their corresponding stochastic filter. We then formally state the
problem we consider, and discuss the uniqueness of its solution. Finally, we outline
a ``direct'' approach to the problem and point to potential computational concerns.

\subsection{Notation}

All vectors are column vectors unless transposed. The vector of all ones is denoted
$\ones$ and the $i$th Cartesian basis vector $e_i$. The element at row $i$ and column $j$
of a matrix is $[\cdot]_{ij}$, and the element at position $i$ of a vector is $[\cdot]_i$.
The vector operator $\diag{\cdot}:\mathbb{R}^n \rightarrow \mathbb{R}^{n\times n}$ gives
the matrix where the vector has been put on the diagonal, and all other elements are zero.
The indicator function $\ind{\cdot}$ takes the value 1 if the expression $\cdot$ is
fulfilled and 0 otherwise. The unit simplex is denoted as $\Delta$.
The nullspace of a matrix is $\ker{}$, and $^\dagger$ denotes pseudo-inverse.

\subsection{Hidden Markov Models}

We refer to a partially observed dynamical model
\eref{eq:hmm_state}-\eref{eq:hmm_observation} whose state space $\mathcal{X} = \{1, \dots,
X\}$ is discrete as a \emph{hidden Markov model} (HMM). We limit ourselves to HMMs with
observation processes on a finite alphabet $\mathcal{Y} = \{1, \dots, Y\}$. 

For such HMMs, the state $x_k$ evolves according to the $X \times X$ transition probability
matrix $P$ with elements 
\begin{equation}
    [P]_{ij} = \Pr{x_{k+1} = j | x_k = i}, \quad i,j \in \mathcal{X}.
\end{equation}
The corresponding observation $y_k$ is generated according to the $X \times Y$ observation
probability matrix $B$ with elements
\begin{equation}
    [B]_{ij} = \Pr{y_k = j | x_k = i}, \quad i \in \mathcal{X}, j \in \mathcal{Y}.
\end{equation}
We denote column $y$ of the observation matrix as $b_y \in \Rb^X$ -- therefore
\begin{equation}
    B = \begin{bmatrix} b_1 & \dots & b_Y \end{bmatrix}.
\end{equation}
Note
that both $P$ and $B$ are row-stochastic matrices; their elements are non-negative and the
elements in each row sum to one.

Under this model structure, it can be shown -- see \cite{krishnamurthy_partially_2016} or
\cite{cappe_inference_2005} for complete treatments -- that the Bayesian filter
\eref{eq:bayesian_filter} for updating the posterior takes the form
\begin{equation}
    \pi_{k} = T(\pi_{k-1}, y_k; P, B) = \frac{\diag{b_{y_{k}}} P^T \pi_{k-1}}{\ones^T \diag{b_{y_{k}}} P^T \pi_{k-1}},
    \label{eq:hmm_filter}
\end{equation}
initialized by $\pi_0$,\footnote{For notational simplicity, we assume that the initial
prior for the filter is the same as the initial distribution of the HMM.} which we refer to
as the \emph{HMM filter}. Here, the posterior $\pi_k \in \Rb^X$ has elements
\begin{equation}
    [\pi_k]_i = \Pr{x_k = i | y_1, \dots, y_k },
\end{equation}
for $i = 1, \dots, X$. Note that $\pi_k \in \{\pi \in \Rb^X : \pi \geq 0, \ones^T \pi
= 1 \} \eqdef \Delta \subset \Rb^X$. That is, the posterior $\pi_k$ lies on the
$(X-1)$-dimensional unit simplex.

\subsection{Inverse Filtering for HMMs}

Although the problems we consider in this paper can be generalized to partially observed
models \eref{eq:hmm_state}-\eref{eq:hmm_observation} on general state and observation
spaces, to obtain tractable algorithms and analytical expressions, we limit ourselves to
only discrete HMMs as introduced in the previous section:
\begin{problem}[Inverse Filtering for HMMs] 

    Given a sequence of posteriors $\pi_0, \pi_1, \dots, \pi_N \in
    \Rb^X$ from an HMM filter \eref{eq:hmm_filter} with known state and observation
    dimensions $X$ and $Y$, reconstruct the following quantities: \emph{i)} the
    transition matrix $P$; \emph{ii)} the observation matrix $B$; \emph{iii)} the
    observations $y_1, \dots, y_N$.

    \label{pr:inverse_filter_discrete}
\end{problem}

To ensure that the problem is well posed, and to simplify our analysis, we make the
following two assumptions:
\begin{assumption}[Ergodicity]
    The transition matrix $P$ and the observation matrix $B$ are elementwise (strictly)
    positive.
    \label{as:P_B_positive}
\end{assumption}

\begin{assumption}[Identifiability] 
    The transition matrix $P$ and the observation matrix $B$ are full column rank.
    \label{as:P_B_full_column}
\end{assumption}

\asref{as:P_B_positive} serves as a proxy for ergodicity of the HMM and the HMM filter --
it is a common assumption in statistical inference for HMMs \cite{baum_statistical_1966,
cappe_inference_2005}. \asref{as:P_B_full_column} is related to identifiability and
assures that no state or observation distribution is a convex combination of that of
another. 

\begin{remark}
    
    Neither of these two assumptions is strict; we violate \asref{as:P_B_positive} in the
    numerical experiments in \sref{sec:numerical}, and \asref{as:P_B_full_column} could be
    relaxed according \cite[Sec.  2.4]{mattila_inverse_2017}. However, they simplify our
    analysis and the presentation.

\end{remark}

\subsection{Identifiability in Inverse Filtering}

Under Assumptions~\ref{as:P_B_positive} and \ref{as:P_B_full_column}, we have the
following identifiability result:
\begin{theorem}
    Suppose that two HMMs $P$, $B$ and $\tilde{P}$, $\tilde{B}$ both satisfy Assumptions
    \ref{as:P_B_positive} and \ref{as:P_B_full_column}. Then the HMM filter is uniquely
    identifiable in terms of the transition and observation matrices. That is, for each observation $y = 1, \dots,
    Y$,
    \begin{equation}
        T(\pi, y; P, B) = T(\pi, y; \tilde{P}, \tilde{B}), \qquad \forall\pi \in \Delta,
        \label{eq:T_equal_T_tilde}
    \end{equation}
    if and only if $P = \tilde{P}$ and $B = \tilde{B}$.

    \label{thrm:Bayesian_filter_unique}
\end{theorem}
\thrmref{thrm:Bayesian_filter_unique} guarantees that the HMM filter update
\eref{eq:hmm_filter} is unique in the sense that two HMMs with different transition and/or
observation matrices cannot generate exactly the same posterior updates. In turn, this
leads us to expect that \pref{pr:inverse_filter_discrete} is well-posed; we should be able
to reconstruct both $P$ and $B$ uniquely once the conditions of
\thrmref{thrm:Bayesian_filter_unique} are fulfilled. 
    
Note, however, that \thrmref{thrm:Bayesian_filter_unique} does not assert that the HMM
filter is identifiable from a sample path of posteriors -- the sample path could be finite
and/or only visit a subset of the probability simplex. In particular, in applying
\thrmref{thrm:Bayesian_filter_unique} for \pref{pr:inverse_filter_discrete}, we would need
to guarantee that updates from every point on the simplex have been observed -- this seems
unnecessarily strong. 

The following theorem is a generalization of \thrmref{thrm:Bayesian_filter_unique} and is
the main identifiability result of this paper. 
\begin{theorem}
    Suppose that two HMMs $P$, $B$ and $\tilde{P}$, $\tilde{B}$ both satisfy Assumptions
    \ref{as:P_B_positive} and \ref{as:P_B_full_column}. Let, for each $y = 1, \dots, Y$, $\Delta_y = \{\pi_1^y, \dots,
    \pi_X^y, \pi_{X+1}^y\} \subset \Delta$ be a set of $X+1$ posteriors such that
    $\pi_1^y, \dots, \pi_X^y \in \Rb^X$ are linearly independent and the last posterior
    $\pi_{X+1}^y \in \Rb^X$ can be written
    \begin{equation}
        \pi_{X+1}^y = [\beta]_1 \pi_1^y + \dots + [\beta]_X \pi_X^y,
    \end{equation}
    with $\beta \in \Rb^X$ and $[\beta]_i \neq 0$ for each $i$. Then, for each 
    $y = 1, \dots, Y$,
    \begin{equation}
        T(\pi, y; P, B) = T(\pi, y; \tilde{P}, \tilde{B}), \qquad \forall \pi \in
        \Delta_y,
        \label{eq:T_equal_T_tilde}
    \end{equation}
    if and only if $P = \tilde{P}$ and $B = \tilde{B}$.

    \label{thrm:Bayesian_filter_unique_subset}
\end{theorem}

\thrmref{thrm:Bayesian_filter_unique_subset} relaxes \thrmref{thrm:Bayesian_filter_unique}
in the sense that, for each observation $y = 1, \dots, Y$, we only need a sequence of $X
+ 1$ posteriors satisfying the conditions for the set $\Delta_y$. To make the theorem more
concrete, for $X = 3$, the conditions mean that all posteriors we have obtained cannot lie
on the lines connecting any two posteriors -- see \fref{fig:condition_for_identifiability} for
an illustration.

\begin{figure}[t!]
    \centering

    \includegraphics{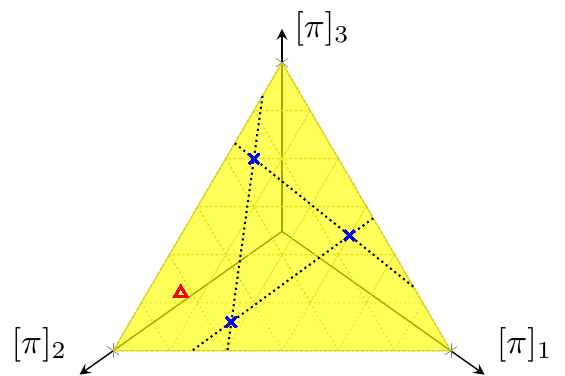}

    \caption{Illustration of the set $\Delta_y$ in
        \thrmref{thrm:Bayesian_filter_unique_subset} (for $X = 3$). The blue crosses
        correspond to $\pi_1^y, \dots, \pi_X^y$ and the red triangle to $\pi_{X+1}^y$. For
        the conditions of the theorem to be fulfilled, the point $\pi_{X+1}^y$ cannot lie
        on the dashed lines.}

    \label{fig:condition_for_identifiability}

\end{figure}

\subsection{Direct Approach to the Inverse Filtering Problem} 

At a first glance, \pref{pr:inverse_filter_discrete} appears computationally intractable:
there are combinatorial elements (due to the unknown sequence of observations) and
non-convexity from the products between columns $b_y$ of the observation matrix and the
transition matrix $P$ in the HMM filter \eref{eq:hmm_filter}. 

In order to reconstruct parameters that are consistent with the data (i.e., that satisfy the
filter equation \eref{eq:hmm_filter} and fulfill the non-negativity and sum-to-one
constraints imposed by probabilities), a direct approach is to solve the
following feasibility problem:
\begin{align}
    \min_{\{y_k\}_{k=1}^N, \, \{b_y\}_{y=1}^Y, \, P} &\quad 
    \sum_{k=1}^N \left\|\pi_{k} - \frac{\diag{b_{y_k}} P^T \pi_{k-1}}{\ones^T \diag{b_{y_k}} P^T \pi_{k-1}}\right\| \notag \\
    \text{s.t.} \quad &\quad y_k \in \{1, \dots, Y\}, \quad\quad \text{ for } k = 1, \dots, N, \notag \\
&\quad b_y \geq 0, \quad\quad\quad\quad\quad\quad \text{ for } y = 1, \dots, Y, \notag \\
&\quad [b_1 \dots b_Y] \ones = \ones, \notag \\
&\quad P \ones = \ones, \quad P \geq 0,
    \label{eq:inverse_filter_naive}
\end{align}
where the choice of norm is arbitrary since the cost is zero for any feasible set of
parameters.


The problem \eref{eq:inverse_filter_naive} is combinatorial (in $\{y_k\}_{k=1}^N\}$) and
non-convex (in $\{b_y\}_{y=1}^Y$ and $P$) -- in other words, it is challenging and
computationally costly to obtain a solution. In the next section, we will propose
an indirect approach that results in a computationally feasible solution to
\pref{pr:inverse_filter_discrete}.

\section{Inverse Filtering by Exploiting \\ the Structure of the HMM Filter}
\label{sec:proposed method}

In this section, we first derive an alternative characterization of the HMM filter
\eref{eq:hmm_filter}. The properties of this characterization allow us to formulate an
alternative solution to \pref{pr:inverse_filter_discrete}. This solution still requires
solving a combinatorial problem (a \emph{nullspace clustering problem}, see
\pref{pr:nullspace_clustering} below) and is, essentially, equivalent to
\eref{eq:inverse_filter_naive}. However, by leveraging insights from its geometrical
interpretation, we derive a computationally feasible convex relaxation based on structured
sparsity regularization (the fused group LASSO \cite{yuan_model_2006}) that permits us to obtain a solution in
a computationally feasible manner.

\subsection{Alternative Characterization of the HMM Filter}

Our first result is a variation of the key result derived in \cite{mattila_inverse_2017}.
First note that the HMM filter \eref{eq:hmm_filter} can be rewritten as 
\begin{equation}
    (\ref{eq:hmm_filter}) \iff b_{y_{k}}^T P^T \pi_{k-1} \pi_{k} = \diag{b_{y_{k}}} P^T
    \pi_{k-1},
    \label{eq:hmm_filter_no_frac}
\end{equation}
by simply multiplying by the denominator (which is allowed under \asref{as:P_B_positive}).
By restructuring\footnote{Detailed algebraic manipulations can be found in the appendix.}
\eref{eq:hmm_filter_no_frac}, we obtain an alternative characterization of the HMM filter
\eref{eq:hmm_filter}:
\begin{theorem}
    Under Assumptions \ref{as:P_B_positive} and \ref{as:P_B_full_column}, the HMM
    filter-update \eref{eq:hmm_filter} can be equivalently written as
    \begin{equation}
        (\pi_{k-1}^T \otimes [\pi_k \ones^T - I]) \vec{\diag{b_{y_k}} P^T} = 0,
        \label{eq:hmm_filter_alternative}
    \end{equation}
    for $k = 1, \dots, N$.
    \label{thrm:hmm_filter_nullspace_PB}
\end{theorem}

To see why the reformulation \eref{eq:hmm_filter_alternative} is useful, recall that in
\pref{pr:inverse_filter_discrete}, we aim to estimate the transition matrix $P$, the
observation matrix $B$ and the observations $y_k$ given posteriors $\pi_k$. Hence, the coefficient matrix
$(\pi_{k-1}^T \otimes [\pi_k \ones^T - I])$ on the left-hand side of
\eref{eq:hmm_filter_alternative} is known to us, and all that we aim to estimate is
contained in its nullspace. 

\subsection{Reconstructing $P$ and $B$ from Nullspaces}
\label{sec:normalizing_into_P_and_B}

It is apparent from \eref{eq:hmm_filter_alternative} that everything we seek to estimate
(i.e., the transition matrix $P$, the observation matrix $B$ and the observations) is
accommodated in a vector that lies in the nullspace of a known coefficient matrix.
Even so, it is not obvious that the sought quantities can be reconstructed from this. In
particular, since a nullspace is only determined up to scalings of its basis vectors, by
leveraging \eref{eq:hmm_filter_alternative} we can at most hope to reconstruct the
\emph{directions} of vectors $\vec{\diag{b_{y}} P^T}$:
\begin{equation}
    \alpha_y \vec{\diag{b_y} P^T} \quad  \in \Rb^{X^2},
    \label{eq:scaled_nullspace_vectors}
\end{equation}
for $y = 1, \dots, Y$, where $\alpha_y \in \Rb_{> 0}$ correspond to scale factors.

Can a set of vectors \eref{eq:scaled_nullspace_vectors} be factorized into $P$ and $B$,
and do the undetermined scale factors $\alpha_y$ (which, again, are due to the nullspace basis
only being determined up to scaling) pose a problem? Our next theorem shows
that it can be done. First, however, note that by reshaping
\eref{eq:scaled_nullspace_vectors}, we equivalently have access to matrices $a_y
\diag{b_y} P^T \in \Rb^{X \times X}$ for $y = 1, \dots, Y$.
\begin{theorem}
    Given are matrices $V_y \eqdef \alpha_y \diag{b_y} P^T$ for $y = 1, \dots, Y$, where $P$
    and $B = \begin{bmatrix}b_1 & \dots & b_Y \end{bmatrix}$ are HMM parameters
    satisfying Assumptions \ref{as:P_B_positive} and \ref{as:P_B_full_column}, and
    $\alpha_y$ are strictly positive (unknown) scalars.  Let 
        $
        \Bysum \eqdef \sum_{y=1}^Y V_y^T,
        $
    then the transition matrix $P$ can be reconstructed as 
    \begin{equation}
        P = \Bysum \diag{\Bysum^{-1} \ones}.
        \label{eq:factorize_to_P}
    \end{equation}
    Subsequently, let $\bar{B} \eqdef \begin{bmatrix} V_1 P^{-T}\ones & \dots & V_Y
    P^{-T}\ones\end{bmatrix}$, then the observation matrix can be reconstructed as
    \begin{equation}
        B = \bar{B} \diag{\bar{B}^\dagger \ones}.
        \label{eq:factorize_to_B}
    \end{equation}
    \label{thrm:reconstruct_P_and_B}
\end{theorem}
The proof of \thrmref{thrm:reconstruct_P_and_B} amounts to algebraically verifying that
the relations \eref{eq:factorize_to_P} and \eref{eq:factorize_to_B} hold by employing
properties of row-stochastic matrices. The last factor in each equation can be interpreted
as a sum-to-one normalization.

\subsection{How to Compute the Nullspaces?} 
\label{sec:how_to_compute_nullspaces}

\thrmref{thrm:reconstruct_P_and_B} gives us a procedure to reconstruct the transition and
observation matrices from vectors \eref{eq:scaled_nullspace_vectors} -- i.e., vectors
parallel to $\vec{\diag{b_y} P^T}$ for $y = 1, \dots, Y$. The goal in the remaining part
of this section is to compute such vectors from the known coefficient matrices in
\eref{eq:hmm_filter_alternative}.

If the nullspace of the coefficient matrix of \eref{eq:hmm_filter_alternative} was
one-dimensional, we could proceed as in \cite{mattila_inverse_2017}: Since there are only
a finite number of values, namely $Y$, that the $y_k$:s can take, there are only a finite
number of directions in which the nullspaces can point; along the vectors $\vec{\diag{b_y} P^T}$ for
$y = 1, \dots, Y$. Hence, once a coefficient matrix corresponding to each observation $y$ has been obtained,
these directions can be reconstructed. 

Unfortunately, the nullspace of the coefficient matrix of \eref{eq:hmm_filter_alternative}
is \emph{not} one-dimensional:
\begin{lemma}
    Under Assumptions \ref{as:P_B_positive} and \ref{as:P_B_full_column}, we have that
    \begin{equation}
        \rank{\pi_{k-1}^T \otimes [\pi_k \ones^T - I]} = X - 1.
    \end{equation}
    \label{lem:dim_of_nullspace}
\end{lemma}
Since $\vec{\diag{b_{y}} P^T} \in \Rb^{X^2}$, the null-space is, in fact, $X^2 - (X - 1)$
dimensional. Below, we demonstrate how, by intersecting multiple nullspaces, we can obtain
a one-dimensional subspace (a vector) that is parallel to the vector $\vec{\diag{b_{y}}
P^T}$ that we seek.

\begin{remark}

    The above is not surprising in the light of that every update \eref{eq:hmm_filter}
    of the posterior corresponds to $X$ equations.\footnote{Actually, only $X-1$ equations
    since the sum-to-one property of the posterior makes one equation superfluous.} In
    \cite{mattila_inverse_2017}, only the $X$ parameters of $\diag{b_y}$ had to be
    reconstructed at each time instant since $P$ was assumed known. Now, instead, we aim to
    reconstruct the $X^2$ parameters of $\diag{b_y} P^T$, which cannot be done with just
    one update (i.e., with $X$ equations). Hence, we will need to employ the equations
    from several updates.
    
    \label{rem:number_of_equations}

\end{remark}

\subsection{Special Case: Known Sequence of Observations}
\label{sec:special_case_known_obs}

To make the workings of our proposed method more transparent, suppose for the moment that
we have access to the sequence of observations $y_1, \dots, y_N$ that were processed by
the filter \eref{eq:hmm_filter}. By \thrmref{thrm:hmm_filter_nullspace_PB}, we know that
the vector $\vec{\diag{b_{y_k}} P^T}$ lies in the nullspace of the coefficient matrix
$(\pi_{k-1}^T \otimes [\pi_k \ones^T - I])$ for all $k = 1, \dots, N$. If we consider only
the time instants when a certain observation, say $y$, was processed, then the vector
$\vec{\diag{b_{y}} P^T}$ lies in the nullspace of all the corresponding coefficient
matrices:
\begin{equation}
    \vec{\diag{b_{y}} P^T} \in 
    \bigcap\limits_{k :\, y_k = y} \ker{(\pi_{k-1}^T \otimes [\pi_k \ones^T - I])},
    \label{eq:vector_in_nullspace}
\end{equation}
for $y = 1, \dots, Y$. Now, if the intersection on the right-hand side of
\eref{eq:vector_in_nullspace} is one-dimensional, this gives us a way to reconstruct
the direction of $\vec{\diag{b_y} P^T}$ -- in this case,
\begin{equation}
    \vspan{\vec{\diag{b_{y}} P^T}} =
    \bigcap\limits_{k :\, y_k = y} \ker{(\pi_{k-1}^T \otimes [\pi_k \ones^T - I])},
    \label{eq:intersecting_nullspaces_known_y}
\end{equation}
and we simply compute the one-dimensional intersection. Recall that the next step would
then be to factorize these directions into the products $P$ and $B$ via
\thrmref{thrm:reconstruct_P_and_B}.

The identifiability result of \thrmref{thrm:Bayesian_filter_unique_subset} tells us that
this happens when there exists a subsequence $\Delta_y \subset \{ \pi_k \}_{k : y_k = y}$, that
satisfies the criteria in \thrmref{thrm:Bayesian_filter_unique_subset}.  

\begin{remark} 

    In practice, roughly $X$ updates, for each $y$, should be enough (see
    Remark~\ref{rem:number_of_equations}). A pessimistic estimate of how many samples will
    be required for this can be obtain via similar reasoning as in \cite[Lemma
    3]{mattila_inverse_2017}: if $B \geq \beta > 0$ elementwise, then on the order of
    $\beta^{-1} YX^2$ samples are expected to suffice. 

    \label{rem:expected_samples}

\end{remark}

\subsection{Inverse Filtering via Nullspace Clustering}
\label{sec:general_solution_nullspace}

\pref{pr:inverse_filter_discrete} is complicated by the fact that in the inverse filtering
problem, we do \emph{not} have access to the sequence of observations -- we only observe
a sequence of posteriors. Thus, we do not know at what time instants a certain observation
$y$ was processed and, hence, which nullspaces to intersect, as in
\eref{eq:intersecting_nullspaces_known_y}, to obtain a vector parallel to $\vec{\diag{b_y}
P^T}$.

An abstract version of the problem can be posed as follows:
\begin{problem}[Nullspace Clustering]

    Given is a set of matrices $\{A_k\}_{k=1}^N$ that can be divided into $Y$ subsets
    (clustered) such that the intersection of the nullspaces of the matrices in each
    subset is one-dimensional. That is, there are numbers $\{y_k\}_{k=1}^N$ with $y_k \in
    \{1, \dots, Y\}$ such that,
    \begin{equation}
        \bigcap\limits_{k :\, y_k = y} \ker{A_k} = \vspan{v_y},
    \end{equation}
    for some vector $v_y$ with $y = 1, \dots, Y$. Find the vectors $\{v_y\}_{y=1}^Y$ that
    span the intersections.

    \label{pr:nullspace_clustering}
\end{problem}
We provide a graphical illustration of the nullspace clustering problem in
\fref{fig:nullspace_clustering}. Note that, in our instantiation of the problem, $A_k
= (\pi_{k-1}^T \otimes [\pi_k \ones^T - I])$ is the coefficient matrix of the HMM filter
\eref{eq:hmm_filter_alternative}, and each vector $v_y = \vec{\diag{b_y} P^T}$ is what we
aim to reconstruct.

\begin{figure}[t!]
    \centering

    \includegraphics{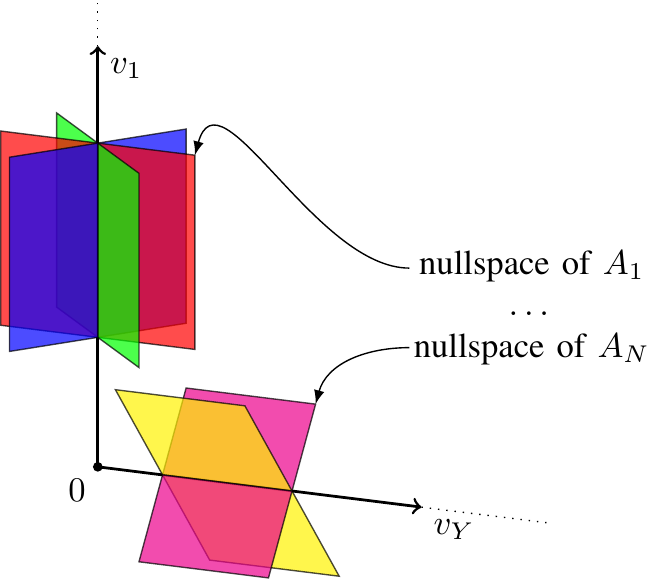}

\caption{We are given a set of subspaces parametrized by the nullspaces of matrices $A_1,
    \dots, A_N$. These subspaces have the property that each subspace contains one of the
    $Y$ vectors $v_1, \dots, v_Y$. In the nullspace clustering problem
    (\pref{pr:nullspace_clustering}), the aim is to find the (directions of) vectors
    $v_1, \dots, v_Y$.}

    \label{fig:nullspace_clustering}

\end{figure}

The problem was simplified in the previous section, since by knowing the observations
$y_1, \dots, y_N$, we know which vector each subspace is generated about (i.e., the
subset assignments) and can simply intersect the subspaces in each subset to obtain the
vectors $v_1, \dots, v_Y$. By not having direct access to the sequence of observations,
the problem becomes combinatorial; which nullspaces should be intersected? Albeit
a solution can be obtained via mixed-integer optimization in much the same fashion as in
\eref{eq:inverse_filter_naive} -- since \pref{pr:nullspace_clustering} is merely
a reformulation of the original problem -- such an approach can be highly computationally
demanding.

We propose instead the following two-step procedure that consists of, first, a convex
relaxation, and second, a refinement step using local heuristics. We emphasize that if the
two steps succeed, then there is \emph{nothing approximate} about the solution we obtain
-- the directions of the vectors $v_y$ are obtained exactly.

\subsubsection*{Step 1. Convex Relaxation} Compute a solution to the convex problem

\begin{align}
    \min_{\{w_k\}_{k=1}^N} &\quad \sum_{i=1}^N \sum_{j=1}^N \| w_i - w_j \|_\infty \notag \\
    \text{s.t.} \quad &\quad A_k w_k = 0, \quad \text{ for } k = 1, \dots, N, \notag \\
                      &\quad w_k \geq 1, \hspace{0.8cm} \text{ for } k = 1, \dots, N,
    \label{eq:nullspace_cluster}
\end{align}
which aims to find $N$ vectors that each lie in the nullspace of the corresponding matrix
$A_k$, and that, by the objective function, are promoted to coincide via a fused group
LASSO \cite{yuan_model_2006}. That is, the set $\{w_k\}_{k=1}^N$ is \emph{sparse} in the
number of unique vectors. This is a relaxation because we have dropped the hard constraint
of there only being exactly $Y$ different vectors. The constraint $w_k \geq 1$ assures
that we avoid the trivial case $w_k = 0$ for all $k$.\footnote{In order to relax
    \asref{as:P_B_positive}, this can be replaced by $\ones^T w_k \geq 1$ since some
components of the nullspace might then be zero.} The $\|\cdot\|_\infty$-norm is used for
convenience since problem \eref{eq:nullspace_cluster} can then be reformulated as a linear
program.

\subsubsection*{Step 2. Refinement via Spherical Clustering}

The solution of \eref{eq:nullspace_cluster} does not completely solve our problem for two
reasons: \emph{i)} it is not guaranteed to return precisely $Y$ unique basis vectors, and
\emph{ii)} it does not tell us to which subset the nullspace of each $A_k$ should be
assigned (i.e., we still do not know which nullspaces to intersect). 

In order to address these two points, we perform a local refinement using spherical
k-means clustering \cite{buchta_spherical_2012} on the set of vectors $\{w_k\}_{k=1}^N$
resulting from \eref{eq:nullspace_cluster}. This provides us with a set of $Y$ centroid
vectors, as well as a cluster assignment of each vector $w_k$. We employ the spherical
version of k-means since we seek nullspace basis vectors -- the appropriate distance
measure is angular spread, and not the Euclidean norm employed in standard k-means
clustering. 

Now, the centroid vectors should provide good approximations of the vectors we seek (since
$\{w_k\}_{k=1}^N$ are expected to be spread around the intersections of the nullspaces by
the sparsity promoting objective in \eref{eq:nullspace_cluster}). However, they do not
necessarily lie in any nullspace since their computation is unconstrained. To obtain an
exact solution to the problem, we go through the $w_k$:s assigned to each cluster in order
of distance to the cluster's centroid and intersect the corresponding $A_k$:s' nullspaces
until we obtain a one-dimensional intersection.

It should be underlined that when an intersection that is one-dimensional has been
obtained, the (direction of) the vector $v_y$ has been computed \emph{exactly}. Recall
that in our instantiation of the problem, each vector $v_y = \vec{\diag{b_y} P^T}$, so
that once these have been reconstructed, they can be decomposed into the transition matrix
$P$ and the observation matrix $B$ according to \thrmref{thrm:reconstruct_P_and_B}.

\begin{figure*}[ht!]

    \hspace{0.4cm} \includegraphics{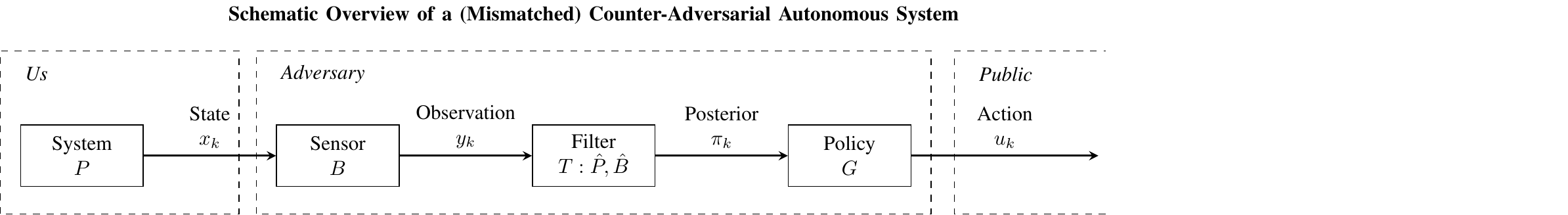}

\caption{Schematic illustration of the setup in
    \eref{eq:game_model1}-\eref{eq:game_model4}. An autonomous adversary measures our
    state $x_k$ as $y_k$ via a sensor $B$.  A mismatched (note that estimates $\hat{P}$
    and $\hat{B}$ are employed) Bayesian filter is used to compute its posterior $\pi_k$
    of our state. Based on this posterior and a policy $G$, a public action $\fact_k$ is
taken. In the remote sensor calibration problem (\pref{pr:remote_calibration}), the aim is
to estimate the adversary's sensor $B$ (which is, in general, different from $\hat{B}$).}

    \label{fig:full_caa_schematic}

\end{figure*}

\begin{remark}

    \thrmref{thrm:reconstruct_P_and_B} assumes that we are given the matrices $V_y$ sorted
    according to the actual labeling of the HMM's observations. If we use the method
    described above, the vectors are only obtained up to permutations of the observation
    labels (this corresponds to the label assigned to each cluster in the spherical
    k-means algorithm). Hence, in practice, we will obtain $B$ up to permutations of its
    columns.

\end{remark}

\subsection{Summary of Proposed Algorithm} 

For convenience, the complete procedure for solving \pref{pr:inverse_filter_discrete} is
summarized in \algoref{alg:full_inverse_filtering}. It should be pointed out that the
algorithm will fail to determine a solution to \pref{pr:inverse_filter_discrete} if it can
not intersect down to a one-dimensional subspace for some $y$. Then, the direction of the
vector $\vec{\diag{b_y} P^T}$ cannot be determined uniquely, and the full set of these
vectors is required in \thrmref{thrm:reconstruct_P_and_B}. This is because this
observation has not been measured enough times, or that the convex relaxation has failed.
As noted in Remark~\ref{rem:number_of_equations}, we expect on the order of $YX^2$
posteriors to be needed for the procedure to succeed -- we explore this further in the
numerical examples in \sref{sec:numerical}.

\begin{remark}

    Once $P$ and $B$ have been reconstructed via \algoref{alg:full_inverse_filtering}, to
    obtain the sequence of observations, simply check which observation $y \in
    \mathcal{Y}$ that maps $\pi_{k-1}$ to $\pi_k$ via the HMM filter \eref{eq:hmm_filter}
    for $k = 1, \dots, N$. This can be done in linear time (in $N$).

    \label{rem:obtaining_observations}
\end{remark}

\begin{algorithm}[b]
\caption{Inverse Filtering (Solution to \pref{pr:inverse_filter_discrete})}
\begin{algorithmic}[1]
  \scriptsize
  \REQUIRE Sequence of posteriors $\{\pi_k\}_{k=1}^N$, dimension $Y$ 
    \STATE Compute the coefficient matrices $A_k = (\pi_{k-1}^T \otimes [\pi_k \ones^T
        - I])$ of \eref{eq:hmm_filter_alternative} for $k = 1, \dots, N$.
    \STATE Compute a solution $\{w_k\}_{k=1}^N$ to the convex problem
        \eref{eq:nullspace_cluster}.
    \STATE Run spherical k-means clustering for $Y$ clusters on the vectors
    $\{w_k\}_{k=1}^N$
    \FOR{each of the $Y$ clusters}
        \STATE $k$-set = $\{ \}$
        \FOR{each $w_k$ in order of increasing distance to its cluster's centroid}
            \STATE Compute intersection of current, and past, corresponding $A_k$:s'
            nullspaces: \\ 
            $\quad$ \emph{i)} Add $k$ to $k$-set, \\
            $\quad$\emph{ii)} Compute $\bigcap_{k \in k\text{-set}} \ker{A_k}$

            \IF{intersection is one-dimensional}
                \STATE Save as $v_y$ and proceed to the next cluster.
            \ENDIF
        \ENDFOR
    \ENDFOR
    \STATE Factorize $\{v_y\}_{y=1}^Y$ into $P$ and $B$ using equations
        \eref{eq:factorize_to_P} and \eref{eq:factorize_to_B}, respectively.
        \STATE To obtain the corresponding sequence of observations, see
        Remark~\ref{rem:obtaining_observations}.
\end{algorithmic}
\label{alg:full_inverse_filtering}
\end{algorithm}

\section{Applications of Inverse Filtering to Counter-Adversarial Autonomous Systems}
\label{sec:application_caa}

During the last decade, the importance of defense against cyber-adversarial and autonomous
treats has been highlighted on numerous occasions -- e.g., \cite{kuptel_counter_2017,
barni_coping_2013, farwell_stuxnet_2011}. In this section, we illustrate how the results
in the previous section can be generalized when \emph{i)} the posteriors from the Bayesian
filter are observed in noise, and \emph{ii)} the filtering system is mismatched. The
problem is motived by remotely calibrating (i.e., estimating) the sensors of an autonomous
adversary by observing its actions. 

\subsection{Counter-Adversarial Autonomous Systems}

Consider an adversary that employs an autonomous filtering and control system that
estimates our state and takes actions based on a control policy.  The goal in the design
of a \emph{counter-adversarial autonomous} (CAA) system is to infer information private to
the adversary, and to predict and guard against its future actions
\cite{kuptel_counter_2017, krishnamurthy_how_2019, mattila_what_2020}.

Formally, it can be interpreted as a two-player game in the form of a \emph{partially observed
Markov decision process} (POMDP; \cite{krishnamurthy_partially_2016}), where information
is partitioned between two players: \emph{us} and the \emph{adversary}. 
The model \eref{eq:hmm_state}-\eref{eq:hmm_observation} is now generalized to:
\begin{align}
    \text{us:}\quad\! x_k &\sim P_{x_{k-1}, x} = p(x | x_{k-1}), \, x_0 \sim \pi_0 \label{eq:game_model1} \\
    \text{adversary:}\quad\! y_k &\sim B_{x_k, y} = p(y | x_k), \\
    \text{adversary:}\quad\! \pi_k &= T(\pi_{k-1}, y_k; \hat{P}, \hat{B}), & \label{eq:game_model3} \\
    \text{adversary \& us:}\quad\! \fact_k &\sim \rpolicy_{\pi_k, \fact} = p(\fact | \pi_k), & \label{eq:game_model4}
\end{align}
which should be interpreted as follows. The state $x_k \in \mathcal{X}$, with initial
condition $\pi_0$, is \emph{our} state that we use to probe the adversary. The observation
$y_k \in \mathcal{Y}$ is made by the \emph{adversary}, who subsequently computes its
posterior (in this setting, we refer to it also as a \emph{belief}) $\pi_k$ of our state
using the Bayesian filter $T$ from \eref{eq:bayesian_filter}.\footnote{Again, for
notational simplicity, we assume that the initial prior for the filter is the same as the
initial distribution of the state.} Note that the adversary does \emph{not} have perfect
knowledge of our transition kernel $P$ nor its sensor $B$; it uses estimates $\hat{P}$ and
$\hat{B}$ in \eref{eq:game_model3}.  Finally, the adversary takes an action $\fact_k \in
\factset$, where $\factset$ is an action set, according to a control policy $\rpolicy$
based on its belief. 

A schematic overview is drawn in \fref{fig:full_caa_schematic}, where the dashed boxes
demarcate information between the players (public means both us and the adversary have
access).

\subsection{Remote Calibration of an Adversary's Sensors}

Various questions can be asked that are of importance in the design of a CAA system. The
specific problem we consider is that of remotely calibrating the adversary's sensor:

\begin{problem}[Remote Calibration of Sensors in CAA Systems] 

    Consider the CAA system \eref{eq:game_model1}-\eref{eq:game_model4}. Given knowledge
    of our realized state sequence $x_0, x_1, x_2, \dots$, our transition kernel $P$ and the
    actions taken by the adversary $\fact_1, \fact_2, \dots$; estimate the observation
    likelihoods $B$ of the adversary's sensor.

    \label{pr:remote_calibration}
\end{problem}

A few remarks on the above problem: Our final targets are the likelihoods $B$, \emph{not}
the adversary's own estimate $\hat{B}$. Previous work \cite{mattila_inverse_2017,
krishnamurthy_how_2019, mattila_what_2020} that considered the above problem, or
variations thereof, assumed that the adversary's filter was perfectly matched (i.e.,
$\hat{P} = P$ and $\hat{B} = B$) and, hence, that $P$ was known to the adversary. We
generalize to the more challenging setup of a mismatched filtering system, and will,
hence, have to estimate the adversary's estimate of our transition kernel. Albeit outside
the scope of this paper, with a solution to \pref{pr:remote_calibration} in place,
a natural extension is the input design problem adapted to CAA systems: How to design our
transition kernel $P$ so as to \emph{i)} maximally confuse the adversary, or \emph{ii)}
optimally estimate its sensor?

In order to connect with the results of the previous section, we consider only discrete
CAA systems -- that is, where the state space $\mathcal{X} = \{1, \dots, X\}$ and observation
space $\mathcal{Y} = \{1, \dots, Y\}$ are discrete. Moreover, we assume that the
dimensions $X$ and $Y$ are known to both us and the adversary.

\subsection{Reconstructing Beliefs from Actions}

The feasibility of \pref{pr:remote_calibration} clearly depends on the adversary's policy
-- for example, if the policy is independent of its belief $\pi_k$, we can hardly hope to
estimate anything regarding its sensors. A natural assumption is that the adversary is
rational and that its policy $\rpolicy$ is based on optimizing its expected
cost\cite{machina_choice_1987, mas_microeconomic_1995, luenberger_microeconomic_1995}:
\begin{align}
    \min_{\fact_k} \quad & \mathbb{E}_{x_k}\left\{c(x_k, \fact_k) \, | \, y_1, \dots, y_k\right\} \notag \\
    \text{s.t.}\quad& \fact_k \in \mathcal{C},
    \label{eq:myopic_adversary}
\end{align}
where $c(x,\fact)$ is a cost function that depends on our state and an action $\fact \in
\mathcal{C} \subset \factset$, with $\mathcal{C}$ a constraint set. That is, the adversary
selects an action by minimizing the cost it is expected to receive in the next time
instant.

Recall that the results in \sref{sec:proposed method} reconstruct filter parameters from
posteriors. In order to leverage these results for \pref{pr:remote_calibration}, we first
need to obtain the adversary's posterior distributions from its actions. This is discussed
to a longer extent in \cite{krishnamurthy_how_2019,
mattila_what_2020}, in a Bayesian framework, and in \cite{mattila_estimating_2019}
in an analytic setting. We will use, and briefly recap, the main results of
\cite{mattila_estimating_2019} below. 

However, first of all, even with a structural form such as \eref{eq:myopic_adversary} in
place, the set of potential policies is still infinite. Without any prior assumptions on
the adversary's preferences and constraints, it is impossible to conclusively infer
specifics regarding its posteriors. In \cite{mattila_estimating_2019}, it is assumed that:
\begin{assumption}

    We know the adversary's cost function $c(x, \fact)$ and its constraint set
    $\mathcal{C}$. Moreover, $c(x,\fact)$ is convex and differentiable in $\fact$. 

    \label{as:known_cost_function}

\end{assumption}

Under this assumption, the full set of posteriors that the adversary could have had at any
time instant was characterized in \cite{mattila_estimating_2019} using techniques from
inverse optimization \cite{iyengar_inverse_2005}. Some regularity conditions are needed to
guarantee that a unique posterior can be reconstructed -- in general, several posteriors
could result in the same action, which would complicate our upcoming treatment of
\pref{pr:remote_calibration}. One set of such conditions is the following:
\begin{assumption}
    The adversary's decision is unconstrained in Euclidean space, i.e., $\mathcal{C}
    = \factset = \Rb^\factdim$ for some dimension $\factdim$. The matrix
    \begin{equation}
        F(\fact) =
        \begin{bmatrix}
            \nabla_\fact c(1, \fact) & \dots & \nabla_\fact c(X, \fact) \\
            1 & \dots & 1
        \end{bmatrix}
        \label{eq:F_u_definition}
    \end{equation}
    has full column rank when evaluated at the observed actions $\fact_1, \dots, \fact_N$.
    \label{as:rank_of_cost_function}
\end{assumption}
This assumption says, roughly, that the cost functions defined in
\eref{eq:myopic_adversary} are ``different enough'', and that no information is truncated
via active constraints. The key result is then the following:
\begin{theorem}

    Under Assumptions~\ref{as:known_cost_function} and \ref{as:rank_of_cost_function}, the
    posteriors of an adversary selecting actions according to \eref{eq:myopic_adversary}
    can be uniquely reconstructed by us from its actions as
    \begin{equation}
        \pi_k = F(\fact_k)^\dagger \begin{bmatrix} 0 & \dots & 0 & 1 \end{bmatrix}^T, 
    \end{equation}
    for $k = 1, \dots, N$, where the matrix $F(u)$ is defined in \eref{eq:F_u_definition}
    and the last vector consists of $X$ zeros and a single one.

    \label{thrm:reconstruct_belief_from_action}

\end{theorem}
The theorem follows directly from \cite[Theorem 1]{mattila_estimating_2019} and the linear
independence of columns of the matrix $F(u)$ in \eref{eq:F_u_definition} -- details are
available in the appendix.

\subsection{Solution to the Remote Sensor Calibration Problem}
\label{sec:solution_to_remote_calibration}

We will now outline a solution to \pref{pr:remote_calibration} that leverages the inverse
filtering algorithms from \sref{sec:proposed method}.

\subsubsection*{Step 1. Reconstruct Posteriors}

Using \thrmref{thrm:reconstruct_belief_from_action}, reconstruct the adversary's sequence
of posteriors $\pi_1, \dots, \pi_N$ from its observed actions and the structural form of
its policy.

\subsubsection*{Step 2. Reconstruct $\hat{P}$ and $\hat{B}$} 

Apply \algoref{alg:full_inverse_filtering} from \sref{sec:proposed method} on the sequence
of posteriors. Note that the posteriors were computed by the adversary using a mismatched
filter \eref{eq:game_model3} -- i.e., as $T(\pi, y; \hat{P}, \hat{B})$ --, so that
\algoref{alg:full_inverse_filtering} reconstructs the adversary's estimates $\hat{P}$ and
$\hat{B}$ of our transition matrix $P$ and its sensor $B$.

\subsubsection*{Step 3. Reconstruct the Observations}

As mentioned in Remark~\ref{rem:obtaining_observations}, once the filter's parameters are
known it is trivial to reconstruct the corresponding sequence of observations $y_1, \dots,
y_N$. 

\subsubsection*{Step 4. Calibrate the Adversary's Sensor}

We now have access to the observations $y_1, \dots, y_N$ that were realized by the HMM
$(P, B)$ and, by the setup of the CAA system \eref{eq:game_model1}-\eref{eq:game_model4},
the corresponding state sequence $x_1, \dots, x_N$. With this information, we can compute
\emph{our} maximum likelihood estimate $\check{B}$ of the adversary's sensor $B$ via
\begin{equation}
    [\check{B}]_{ij} = \frac{\sum_{k=1}^N \ind{x_k = i, y_k = j}}{\sum_{k=1}^N \ind{x_k = i}},
    \label{eq:ml_estimate_known_x}
\end{equation}
which corresponds to the M-step in the \emph{expectation-maximization} (EM) algorithm for
HMMs -- see, e.g., \cite[Section 6.2.3]{vidyasagar_hidden_2014}. This completes the
solution to \pref{pr:remote_calibration}.

\subsubsection*{Discussion}

It is worth making a few remarks at this point. First of all, it should be underlined that
$\hat{B} \neq \check{B}$ -- that is, our estimate $\check{B}$ is not necessarily equal to
that of the adversary $\hat{B}$. In fact, our estimate depends on the number $N$ of observed
actions and is, as such, \emph{improving} over time by the asymptotic consistency of
the maximum likelihood estimate. If the adversary does not recalibrate its estimate
$\hat{B}$ online, then for large enough $N$, our estimate will eventually \emph{be more
accurate} than the adversary's own estimate.

Moreover, the steps in \sref{sec:solution_to_remote_calibration} are independent of the
accuracy of the adversary's estimates $\hat{P}$ and $\hat{B}$ (as long as they fulfill
Assumptions~\ref{as:P_B_positive} and \ref{as:P_B_full_column}) since we are exploiting
the algebraic structure of its filter. This means that, even if the adversary employs a bad
estimate of our transition matrix $P$, as long as it is taking actions, we can improve our
estimate of its sensor $B$.

Finally, if the setup is modified so that we do not have access to the transition matrix
$P$ or the realized state sequence in \eref{eq:game_model1}, then the step
\eref{eq:ml_estimate_known_x} would be replaced by the full EM algorithm. This would
compute our estimates $\check{P}$ and $\check{B}$ of the system's transition matrix $P$
and the adversary's sensor $B$ based on the reconstructed observations made by the
adversary.  Again, by the asymptotic properties of the maximum likelihood estimate, the
accuracy of these estimates (that depend on $N$) will eventually surpass those of the
adversary.

\section{Numerical Examples}
\label{sec:numerical}

In this section, we evaluate and illustrate the proposed inverse filtering algorithms in
numerical examples. All simulations were run in MATLAB R2018a on a 1.9 GHz CPU.

\subsection{Reconstructing $P$ and $B$ via Inverse Filtering}

\begin{figure}[t!] 

    \centering
    \includegraphics{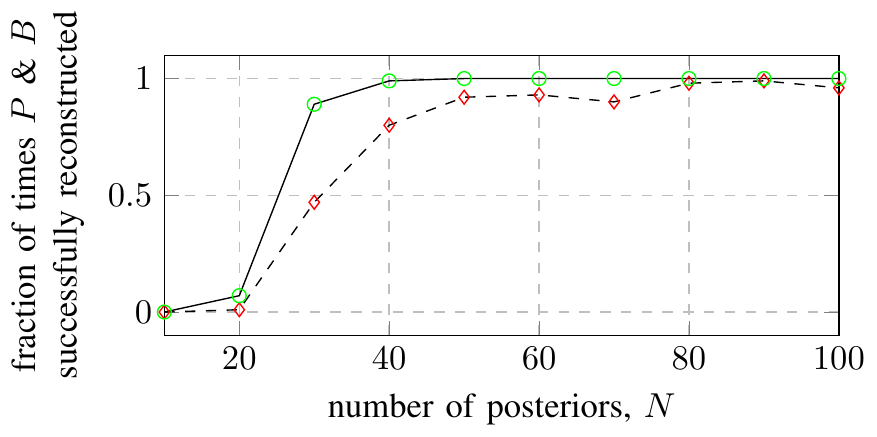}

    \caption{In \pref{pr:inverse_filter_discrete}, we obtain a sequence
        $\{\pi_k\}_{k=1}^N$ of posteriors from an HMM filter \eref{eq:hmm_filter}. Out of
        100 realizations, for each value of $N$, we compute the fraction of times that
        \algoref{alg:full_inverse_filtering} successfully reconstructs the transition
        matrix $P$ and the observation matrix $B$ of the HMM (red diamonds). We also plot
        the success rate of an oracle method (green circles) that has access to the
        observations. With around $N = 50$ posteriors, the algorithm succeeds frequently.}

    \label{fig:success_probability_random_walk}
\end{figure}

Recall that \pref{pr:inverse_filter_discrete} aims to reconstruct HMM parameters given
a sequence of posterior distributions. \algoref{alg:full_inverse_filtering} is
deterministic, but there is randomness in terms of the data (the realization of the HMM)
which can cause the algorithm to fail to reconstruct the HMM's parameters. This can happen
for three different reasons: First, if a certain observation has been measured too few
times then there is fundamentally too few equations available to reconstruct the
parameters -- see Remark~\ref{rem:number_of_equations} (in
\sref{sec:how_to_compute_nullspaces}). Second, if too few independent equations have been
generated, we do not have identifiability and cannot intersect to a one-dimensional
subspace in \eref{eq:intersecting_nullspaces_known_y} -- see
\thrmref{thrm:Bayesian_filter_unique_subset}. Third, we rely on a convex relaxation to
solve the original combinatorial problem. This is a heuristic and it is not guaranteed to
converge to a solution of the original problem. Hence, in these simulations, we estimate
the \emph{probability} of the algorithm succeeding (with respect to the realization of the
HMM data).

In order to demonstrate that the assumptions we have made in the paper are not strict, we
consider the following HMM:
\begin{equation}
    P = \begin{bmatrix}
            0 & \nicefrac{1}{2} & 0 & 0 & \nicefrac{1}{2} \\
            \nicefrac{1}{2} & 0 & \nicefrac{1}{2} & 0 & 0 \\
            0 & \nicefrac{1}{2} & 0 & \nicefrac{1}{2} & 0 \\
            0 & 0 & \nicefrac{1}{2} & 0 & \nicefrac{1}{2} \\
            \nicefrac{1}{2} & 0 & 0 & \nicefrac{1}{2} & 0 \\
    \end{bmatrix}, \;
    B = \begin{bmatrix}
            \nicefrac{2}{5} & \nicefrac{2}{5} & \nicefrac{1}{5} \\
            \nicefrac{2}{5} & \nicefrac{2}{5} & \nicefrac{1}{5} \\
            \nicefrac{2}{5} & \nicefrac{1}{5} & \nicefrac{2}{5} \\
            \nicefrac{1}{5} & \nicefrac{2}{5} & \nicefrac{2}{5} \\
            \nicefrac{1}{5} & \nicefrac{2}{5} & \nicefrac{2}{5} \\
    \end{bmatrix}.
    \label{eq:P_B_random_walk}
\end{equation}
Note that its transition matrix corresponds to a random walk, which violates \asref{as:P_B_positive}.
We consider a reconstruction successful if the error in norm is smaller than $10^{-3}$ for
both $P$ and $B$. We generated 100 independent realizations for a range of values of $N$
(the number of posteriors). The fractions of times the algorithms were successful are
plotted in \fref{fig:success_probability_random_walk} with red diamonds for
\algoref{alg:full_inverse_filtering}, and with green circles for an oracle method that has
access to the corresponding sequence of observations. The oracle method provides an upper
bound on the success rate (if it fails, it is not possible to uniquely reconstruct the HMM
parameters, as discussed above). The gap between the curves is due to the convex relaxation.

\begin{figure}[t!]

    \centering
    \includegraphics{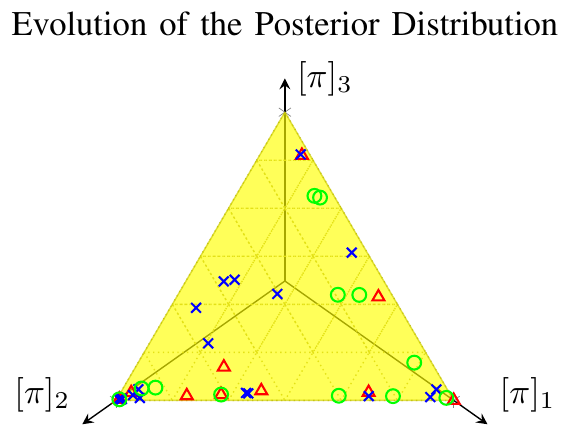}

\caption{A realization of the set $\{\pi_k\}_{k=1}^{50}$, corresponding to HMM
    \eref{eq:HMM_in_numerical_3_3}, illustrated on the simplex. The points are labeled
    according to what observation was measured: $y = 1$ (red triangle), $y = 2$ (blue
    cross), $y = 3$ (green circle). To fulfill the conditions of
\thrmref{thrm:Bayesian_filter_unique_subset}, the points corresponding to each observation
cannot all lie on the lines connecting any two points -- see
\fref{fig:condition_for_identifiability}.}

\label{fig:evolution_of_posteriors}

\end{figure}

A few things should be noted from \fref{fig:success_probability_random_walk}. First, with
only around 50 posteriors from the HMM filter, the fraction of times the algorithm
succeeds in solving \pref{pr:inverse_filter_discrete} is high. It was noted in
Remark~\ref{rem:expected_samples} (in \sref{sec:special_case_known_obs}), that since $B
\geq \nicefrac{1}{5} \eqdef \beta$, a rough estimate is that $\beta^{-1}YX^2 = 375$
posteriors should suffice -- however, as is clear from
\fref{fig:success_probability_random_walk}, this is a very conservative estimate since the
algorithm is very successful already with around 50 posteriors. Second, the gap between
the two curves is small -- hence, the convex relaxation is successful in achieving
a solution to the original combinatorial problem. Finally, it should be mentioned that
with 50 posteriors, the run-time is approximately thirty seconds.

\subsection{Evolution of the Posterior Distribution}

Next, to illustrate the conditions of \thrmref{thrm:Bayesian_filter_unique_subset}, we
plot the set $\{\pi_k\}_{k=1}^{50}$ on the simplex. To be able to visualize the
data, we now consider an HMM with dimension $X = 3$:
\begin{equation}
    P = \begin{bmatrix}
        0.8 & 0.1 & 0.1 \\
        0.05 & 0.9 & 0.05 \\
        0.2 & 0.1 & 0.7
    \end{bmatrix}, \quad
    B = \begin{bmatrix}
        0.7 & 0.1 & 0.2 \\
        0.1 & 0.8 & 0.1 \\
        0.05 & 0.05 & 0.9
    \end{bmatrix}.
    \label{eq:HMM_in_numerical_3_3}
\end{equation}

In \fref{fig:evolution_of_posteriors}, each posterior has been marked according to which
observation was subsequently measured. We can clearly find four posteriors, for each
observation, that are sufficiently disperse (see the illustration in
\fref{fig:condition_for_identifiability}), and hence fulfill the conditions for
\thrmref{thrm:Bayesian_filter_unique_subset} that guarantees a unique solution.

The results of simulations with the same setup as before are shown in
\fref{fig:success_probability}, and similar conclusions can be drawn. First, with around
50 posteriors, the algorithm has a high success rate. Second, the bound in
Remark~\ref{rem:expected_samples} (in \sref{sec:special_case_known_obs}) postulate an
expected number of posteriors of 540, which again is conservative.  

\begin{figure}[t!] 

    \centering
    \includegraphics{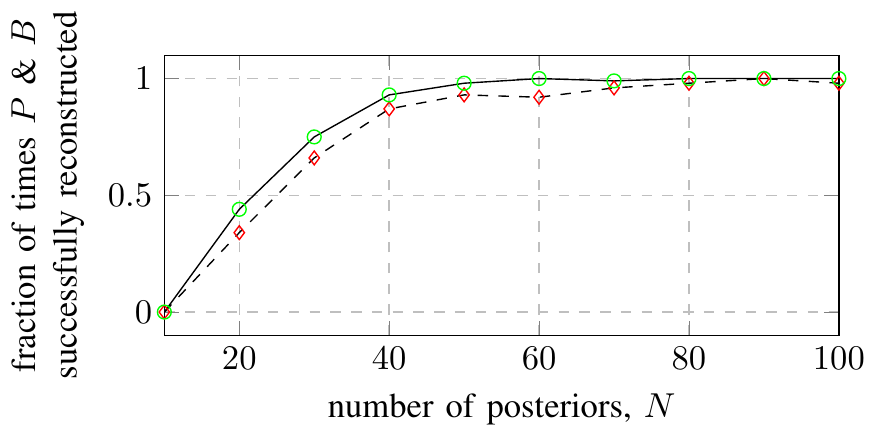}

    \caption{Same setup as in \fref{fig:success_probability_random_walk}, but for the HMM
    \eref{eq:HMM_in_numerical_3_3}. With access to around 50 posteriors,
\algoref{alg:full_inverse_filtering} succeeds in reconstructing $P$ and $B$ very
frequently.}

    \label{fig:success_probability}
\end{figure}

\section{Conclusions and Future work}

In this paper, we considered inverse filtering problems for \emph{hidden Markov models}
(HMMs) with finite observation spaces. We proposed algorithms to reconstruct the
transition kernel, the observation likelihoods and the measured observations from
a sequence of posterior distributions computed by an HMM filter. Our key algorithm is
a two-step procedure based on a convex relaxation and a subsequent local refinement. We
discussed and derived conditions for identifiability in inverse filtering. As an
application of our results, we demonstrated how the proposed algorithms can be employed in
the design of counter-adversarial autonomous systems: How to remotely calibrate the
sensors of an autonomous adversary based on its actions? Finally, the algorithms were
evaluated and verified in numerical simulations.

In the future, we would like to analyze the nullspace clustering problem further and see
if it applies to other settings. We would also like to generalize the setup to other
policy-structures of the adversary -- for example, when its action set is discrete. 

\bibliography{rob_references}

\nobalance
\appendices

\section{Proof of \thrmref{thrm:Bayesian_filter_unique}}

To prove \thrmref{thrm:Bayesian_filter_unique}, we will use the following auxiliary
lemma:
\begin{lemma}
    Let $A \in \Rb^{X \times X}$ and $M \in \Rb^{X \times X}$ be two non-singular
    matrices. If
    \begin{equation}
        Ax = \kappa(x) Mx, \qquad \forall x \in \Delta,
        \label{eq:Ax_alpha_Mx}
    \end{equation}
    where $\kappa(x)$ is a non-zero scalar, and $\Delta = \{ x \in \Rb^X : x \geq
    0, \ones^T x = 1\}$ is the unit simplex, then
    \begin{equation}
        A = \kappa M,
    \end{equation}
    where $\kappa$ is a non-zero constant scalar.
    
    \label{eq:eigenvalue_constant_on_simplex}
\end{lemma}
\begin{proof}
    Consider the $i$th Cartesian basis vector $e_i \in \Delta$:
    \begin{equation}
        A e_i = \kappa(e_i) M e_i.
        \label{eq:ev_equation_for_e_i}
    \end{equation}
    Concatenate \eref{eq:ev_equation_for_e_i} for $i = 1, \dots, X$, to get
    \begin{align}
        A \begin{bmatrix} e_1 & \dots & e_X \end{bmatrix} &= M \begin{bmatrix} \kappa(e_1)
    e_1 & \dots & \kappa(e_X) e_X \end{bmatrix} \implies \notag \\
        A &= M \begin{bmatrix}
            \kappa(e_1) & & 0 \\
                        & \ddots & \\
                      0 & & \kappa(e_X)
        \end{bmatrix}.
        \label{eq:diagonal_with_different_elements}
    \end{align}
    Next, consider any vector on the simplex with non-zero components:
    \begin{equation}
        x = [x]_1 e_1 + \dots + [x]_X e_X \qquad \in \Delta,
    \end{equation}
    such that $[x]_i \neq 0$ for $i = 1, \dots, X$. Introducing
    \eref{eq:diagonal_with_different_elements} in \eref{eq:Ax_alpha_Mx} for this $x$
    yields
    \begin{align}
        M& \begin{bmatrix}
            \kappa(e_1) & & 0 \\
                        & \ddots & \\
                      0 & & \kappa(e_X)
          \end{bmatrix} \left( [x]_1 e_1 + \dots + [x]_X e_X \right) \notag \\
          &= \kappa(x) M \left( [x]_1 e_1 + \dots + [x]_X e_X \right) \implies \notag \\
          \kappa(e_1) &[x]_1 e_1 + \dots + \kappa(e_X) [x]_X e_X \notag \\
                      &\hspace{1.1cm}= \kappa(x) [x]_1 e_1 + \dots + \kappa(x) [x]_1 e_X , 
          \label{eq:kappa_function_kappa_constant}
    \end{align}
    where in the implication we have multiplied by $M^{-1}$ from the left and simplified
    the expression. Since the $e_i$:s are linearly independent, consider any component of
    \eref{eq:kappa_function_kappa_constant}:
    \begin{align}
        \kappa(e_i) [x]_i &= \kappa(x) [x]_i \implies \notag \\
        \kappa(e_i) &= \kappa(x),
    \end{align}
    for $i = 1, \dots, X$, since $[x]_i \neq 0$. In other words,
    \begin{equation}
        \kappa(e_1) = \dots = \kappa(e_X) = \kappa(x) \eqdef \kappa
    \end{equation}
    is constant.
    Introducing this in \eref{eq:diagonal_with_different_elements} yields
    \begin{align}
        A &= M \begin{bmatrix}
            \kappa & & 0 \\
                        & \ddots & \\
                      0 & & \kappa
        \end{bmatrix} \notag \\
        &= \kappa M.
    \end{align}

\end{proof}
To employ \lref{eq:eigenvalue_constant_on_simplex}, we reformulate \eref{eq:T_equal_T_tilde} as follows:
\begin{align}
    T(\pi, y; P, B) &= T(\pi, y; \tilde{P}, \tilde{B}) \implies \notag \\
    \frac{\diag{b_{y}} P^T \pi}{\ones^T \diag{b_{y}} P^T \pi} &=
    \frac{\diag{\tilde{b}_{y}} \tilde{P}^T \pi}{\ones^T \diag{\tilde{b}_{y}} \tilde{P}^T
    \pi} \implies \notag \\
    \diag{b_{y}} P^T \pi&=
    \frac{\ones^T \diag{b_{y}} P^T \pi}{\ones^T \diag{\tilde{b}_{y}} \tilde{P}^T \pi}
    \diag{\tilde{b}_{y}} \tilde{P}^T \pi,
    \label{eq:T_equal_T_tilde_reformulated}
\end{align}
which holds for all $\pi \in \Delta$ and $y = 1, \dots, Y$.

Next, we consider \eref{eq:T_equal_T_tilde_reformulated} for a fixed $y$, and note that the
matrices $\diag{b_{y}} P^T$ and $\diag{\tilde{b}_{y}} \tilde{P}^T$ are non-singular (by
Assumptions~\ref{as:P_B_positive} and \ref{as:P_B_full_column}).
\lref{eq:eigenvalue_constant_on_simplex} then yields that
\begin{align}
    \diag{b_{y}} P^T &= \alpha(y) \diag{\tilde{b}_y} \tilde{P}^T  \implies \notag \\
    P \diag{b_{y}} &= \tilde{P} \alpha(y) \diag{\tilde{b}_y},
    \label{eq:Pdiag_tildeP}
\end{align}
where $\alpha(y) \in \Rb$ is a scalar and which holds for $y = 1, \dots, Y$.

If we sum equations \eref{eq:Pdiag_tildeP} over $y$, and use the fact that $B$ is a
stochastic matrix,
\begin{align}
    \sum_{y=1}^Y P \diag{b_{y}} &= \sum_{y=1}^Y \tilde{P} \alpha(y) \diag{\tilde{b}_y}, & \implies \notag \\
    P &= \tilde{P} \sum_{y=1}^Y \alpha(y) \diag{\tilde{b}_y}, &
    \label{eq:Pdiag_tildeP_sum_over_y}
\end{align}
and then right-multiply by $\ones$, and use that $P$ is a stochastic matrix, we obtain
\begin{align}
    P \ones &= \tilde{P} \sum_{y=1}^Y \alpha(y) \diag{\tilde{b}_y} \ones, & \implies \notag \\
    \ones &= \tilde{P} \sum_{y=1}^Y \alpha(y) \tilde{b}_y. &
\end{align}
Pre-multiplying by $\tilde{P}^{-1}$ yields
\begin{align}
    \tilde{P}^{-1} \ones &= \tilde{P}^{-1} \tilde{P} \sum_{y=1}^Y \alpha(y) \tilde{b}_y & \implies \notag \\
    \ones &= \sum_{y=1}^Y \alpha(y) \tilde{b}_y, & 
    \label{eq:tilde_B_alpha_weigth}
\end{align}
where we have used the fact that the row-sums of the inverse of a (row) stochastic matrix
are all equal to one\footnote{To see this, consider an invertible stochastic matrix $A$: $A\ones
= \ones \implies A^{-1} A \ones = A^{-1} \ones \implies A^{-1}\ones = \ones$.}. By 
applying the $\diag{}$-operation to \eref{eq:tilde_B_alpha_weigth}, we see that
\begin{equation}
    I = \sum_{y=1}^Y \alpha(y) \diag{\tilde{b}_y},
\end{equation}
which when introduced in \eref{eq:Pdiag_tildeP_sum_over_y} yields that 
\begin{equation}
    \tilde{P} = P.
\end{equation}

Next, note that \eref{eq:tilde_B_alpha_weigth} can be rewritten as $\ones = \tilde{B}
\begin{bmatrix} \alpha(1) & \dots & \alpha(Y) \end{bmatrix}^T$. We know that
$\tilde{B} \ones = \ones$, since it is a stochastic matrix, and that $\tilde{B}$ has full
column rank by assumption. Hence, $\ones = \begin{bmatrix} \alpha(1) & \dots & \alpha(Y)
\end{bmatrix}^T$. This yields 
\begin{equation}
    b_y = \tilde{b}_y,
\end{equation}
for $y = 1, \dots, Y$, from \eref{eq:Pdiag_tildeP} by first pre-multiplying by $P^{-1}$.

The other direction is trivial: if $P = \tilde{P}$ and $B = \tilde{B}$, then $T(\pi, y; P,
B) = T(\pi, y; \tilde{P}, \tilde{B})$ for all $\pi \in \Delta$ and each $y = 1, \dots, Y$ by \eref{eq:hmm_filter}:
\begin{equation}
    \frac{\diag{b_{y}} P^T \pi}{\ones^T \diag{b_{y}} P^T \pi} =
    \frac{\diag{\tilde{b}_{y}} \tilde{P}^T \pi}{\ones^T \diag{\tilde{b}_{y}} \tilde{P}^T
    \pi}.
\end{equation}

\section{Proof of \thrmref{thrm:Bayesian_filter_unique_subset}}

We begin by giving a generalization of \lref{eq:eigenvalue_constant_on_simplex}:
\begin{lemma}
    Let $A \in \Rb^{X \times X}$ and $M \in \Rb^{X \times X}$ be two non-singular matrices
    and $\mathcal{Z} = \{z_1, \dots, z_X, z_{X+1}\}$ be a set of vectors where $z_1,
    \dots, z_X \in \Rb^X$ are linearly independent and the last vector when expressed in
    this basis has non-zero components -- that is, $z_{X+1} \in \Rb^X$ can be written
    \begin{equation}
        z_{X+1} = [\beta]_1 z_1 + \dots + [\beta]_X z_X,
        \label{eq:express_z_with_beta}
    \end{equation}
    with $\beta \in \Rb^X$ and $[\beta]_i \neq 0$ for $i = 1, \dots, X$. If
    \begin{equation}
        Ax = \kappa(x) Mx, \qquad \forall x \in \mathcal{Z},
        \label{eq:Ax_alpha_Mx_general}
    \end{equation}
    where $\kappa(x)$ is a non-zero scalar, then
    \begin{equation}
        A = \kappa M,
    \end{equation}
    where $\kappa$ is a non-zero constant scalar.

    \label{eq:eigenvalue_constant_on_simplex_general}
\end{lemma}
\begin{proof}

    For the $i$th vector in $\mathcal{Z}$, we have by \eref{eq:Ax_alpha_Mx_general} that 
    \begin{equation}
        A z_i = \kappa(z_i) M z_i,
    \end{equation}
    which can be concatenated to
    \begin{align}
        A \begin{bmatrix} z_1 & \dots & z_X \end{bmatrix} &= M \begin{bmatrix} \kappa(z_1)
    z_1 & \dots & \kappa(z_X) z_X \end{bmatrix} \implies \notag \\
        AZ &= M  \begin{bmatrix} z_1 & \dots & z_X \end{bmatrix}
        \begin{bmatrix}
            \kappa(z_1) & & 0 \\
                        & \ddots & \\
                      0 & & \kappa(z_X)
        \end{bmatrix} \notag \\
        &= MZ \begin{bmatrix}
            \kappa(z_1) & & 0 \\
                        & \ddots & \\
                      0 & & \kappa(z_X)
        \end{bmatrix},
        \label{eq:A_applied_to_Z}
    \end{align}
    where we have denoted $Z = \begin{bmatrix} z_1 & \dots & z_X \end{bmatrix}$.
    
    We can rewrite \eref{eq:express_z_with_beta} with this definition of $Z$ as
    \begin{equation}
        z_{X+1} = Z\beta,
        \label{eq:z_X1_matrix_vector}
    \end{equation}
    which together with \eref{eq:A_applied_to_Z} yields
    \begin{align}
        A z_{X+1} &= AZ\beta \notag \\
                  &= MZ \begin{bmatrix}
                            \kappa(z_1) & & 0 \\
                                        & \ddots & \\
                                      0 & & \kappa(z_X)
                        \end{bmatrix} \beta.
        \label{eq:AzX1_first}
    \end{align}
    Next, by employing \eref{eq:Ax_alpha_Mx_general} for $z_{X+1}$ and using
    \eref{eq:z_X1_matrix_vector}, we obtain
    \begin{align}
        A z_{X+1} &= \kappa(z_{X+1}) M z_{X+1} \notag \\
                  &= \kappa(z_{X+1}) M Z \beta.
        \label{eq:AzX1_second}
    \end{align}
    Equating \eref{eq:AzX1_first} and \eref{eq:AzX1_second} yields
    \begin{align}
        MZ \begin{bmatrix}
            \kappa(z_1) & & 0 \\
                        & \ddots & \\
                      0 & & \kappa(z_X)
      \end{bmatrix} \beta &= \kappa(z_{X+1}) M Z \beta \implies \notag \\
        \begin{bmatrix}
            \kappa(z_1) & & 0 \\
                        & \ddots & \\
                      0 & & \kappa(z_X)
      \end{bmatrix} \beta &= \kappa(z_{X+1}) \beta,
      \label{eq:beta_two_expressions}
    \end{align}
    by multiplying by the inverse of $MZ$ from the left. The $i$th component of
    \eref{eq:beta_two_expressions} is
    \begin{align}
        \kappa(z_i) [\beta]_i &= \kappa(z_{X+1}) [\beta]_i \implies \notag \\
        \kappa(z_i) &= \kappa(z_{X+1}),
    \end{align}
    since $[\beta]_i \neq 0$. Hence,
    \begin{equation}
        \kappa(z_1) = \dots = \kappa(z_X) = \kappa(z_{X+1}) \eqdef \kappa,
    \end{equation}
    which when introduced in \eref{eq:A_applied_to_Z} yields
    \begin{align}
        AZ &= MZ \begin{bmatrix}
            \kappa & & 0 \\
                        & \ddots & \\
                      0 & & \kappa
        \end{bmatrix} \notag \\
        &= \kappa MZ,
    \end{align}
    or, finally, $A = \kappa M$ by multiplying with the inverse of $Z$ from the right.

\end{proof}

\begin{remark}

    Note that \lref{eq:eigenvalue_constant_on_simplex} follows from
    \lref{eq:eigenvalue_constant_on_simplex_general} by considering the vectors $e_1,
    \dots, e_X \in \Delta$ and any vector in the interior of the simplex.

\end{remark}

As in the proof of \thrmref{thrm:Bayesian_filter_unique}, to employ
\lref{eq:eigenvalue_constant_on_simplex_general}, we first reformulate
\eref{eq:T_equal_T_tilde} as follows:
\begin{align}
    T(\pi, y; P, B) &= T(\pi, y; \tilde{P}, \tilde{B}) \implies \notag \\
    \frac{\diag{b_{y}} P^T \pi}{\ones^T \diag{b_{y}} P^T \pi} &=
    \frac{\diag{\tilde{b}_{y}} \tilde{P}^T \pi}{\ones^T \diag{\tilde{b}_{y}} \tilde{P}^T
    \pi} \implies \notag \\
    \diag{b_{y}} P^T \pi&=
    \frac{\ones^T \diag{b_{y}} P^T \pi}{\ones^T \diag{\tilde{b}_{y}} \tilde{P}^T \pi}
    \diag{\tilde{b}_{y}} \tilde{P}^T \pi,
\end{align}
which holds for $y = 1, \dots, Y$ and $\pi \in \Delta_y$.

For a fixed $y$, the conditions of \lref{eq:eigenvalue_constant_on_simplex_general} are
fulfilled; identify $\Delta_y$ with the set $\mathcal{Z}$ and note that the matrices
$\diag{\tilde{b}_{y}} \tilde{P}^T$ and $\diag{\tilde{b}_{y}} \tilde{P}^T$ are
non-singular. Hence, 
\begin{equation}
    \diag{b_{y}} P^T = \alpha(y) \diag{\tilde{b}_y} \tilde{P}^T,
\end{equation}
or, by taking the transpose,
\begin{equation}
    P \diag{b_{y}} = \tilde{P} \alpha(y) \diag{\tilde{b}_y},
\end{equation}
for $y = 1, \dots, Y$. The setup is now exactly the same as after \eref{eq:Pdiag_tildeP}
in the proof of \thrmref{thrm:Bayesian_filter_unique} -- the rest of the proof is
identical.

\section{Proof of \thrmref{thrm:hmm_filter_nullspace_PB}}

Multiply expression \eref{eq:hmm_filter} for the HMM filter by its denominator to
obtain \eref{eq:hmm_filter_no_frac}, and then reshuffle the terms:
\begin{align}
    \ones^T \diag{b_{y_{k}}} P^T \pi_{k-1} \pi_{k} = \diag{b_{y_{k}}} P^T \pi_{k-1} & \iff \notag \\
    \pi_{k} \ones^T \diag{b_{y_{k}}} P^T \pi_{k-1} = \diag{b_{y_{k}}} P^T \pi_{k-1} & \iff \notag \\
        \left( \pi_{k} \ones^T - I \right) \diag{b_{y_{k}}} P^T \pi_{k-1} = 0. &
    \label{eq:nullspace_equation_before_vec}
\end{align}
By vectorizing and applying a well-known result relating the vectorization operator to
Kronecker products\cite{horn_topics_1991}, with an appropriate grouping of the terms, we
obtain that
\begin{align}
    \vec{\left[ \pi_{k} \ones^T - I \right] \left( \diag{b_{y_{k}}} P^T \right) \pi_{k-1}} = 0 & \iff \notag \\
    \left( \pi_{k-1}^T \otimes [\pi_{k} \ones^T - I] \right) \vec{\diag{b_{y_{k}}} P^T} = 0. &
\end{align}

\section{Proof of \thrmref{thrm:reconstruct_P_and_B}}

By definition, we have that
\begin{align}
    \Bysum &\eqdef \sum_{y=1}^Y V_y^T \notag \\
           &= \sum_{y=1}^Y (\alpha_y \diag{b_y} P^T)^T \notag \\
           &= P \sum_{y=1}^Y \alpha_y \diag{b_y}.
    \label{eq:Bysum_explicit}
\end{align}
First, note that $\Bysum$ is invertible since $P$ is invertible (by
\asref{as:P_B_full_column}) and the result of the summation is a diagonal matrix with strictly
positive entries (by \asref{as:P_B_positive}). Next, we evaluate $\Bysum \diag{\Bysum^{-1}
\ones}$ by introducing \eref{eq:Bysum_explicit}:
\begin{align}
    \Bysum & \diag{\Bysum^{-1} \ones} \notag \\
           &= P \left( \sum_{y=1}^Y \alpha_y \diag{b_y} \right) \diag{(P \sum_{y=1}^Y \alpha_y
    \diag{b_y})^{-1} \ones} \notag \\
    &= P \left( \sum_{y=1}^Y \alpha_y \diag{b_y} \right) \diag{(\sum_{y=1}^Y \alpha_y \diag{b_y})^{-1} P^{-1} \ones} \notag \\
    &= P \left( \sum_{y=1}^Y \alpha_y \diag{b_y} \right) \diag{(\sum_{y=1}^Y \alpha_y \diag{b_y})^{-1} \ones} \notag \\
    &= P \left( \sum_{y=1}^Y \alpha_y \diag{b_y} \right) \left( \sum_{y=1}^Y \alpha_y \diag{b_y} \right)^{-1} \notag \\
    &= P,
\end{align}
where in the third equality we used the fact that the inverse of a row-stochastic matrix
has elements on each row that sum to one\footnote{Assume $A$ is invertible and
row-stochastic: $A \ones = \ones \implies A^{-1} A \ones = A^{-1} \ones \implies \ones
= A^{-1} \ones$.}, and in the fourth that the result of the summation is a diagonal matrix
and that it has a diagonal inverse that is obtained by inverting each element. This allows
us to reconstruct the transition matrix $P$.

To reconstruct the observation matrix, we proceed as follows. First note that by
multiplying $V_y$ by $P^{-T} \ones$ from the right, we obtain
\begin{align}
    V_y P^{-T} \ones &= \alpha_y \diag{b_y} P^T (P^{-T} \ones) \notag \\
                     &= \alpha_y b_y,
\end{align}
which is column $y$ of the observation matrix scaled by a factor $\alpha_y$. By
horizontally stacking such vectors, we build the matrix
\begin{align}
    \bar{B} &\eqdef \begin{bmatrix} V_1 P^{-T}\ones & \dots & V_Y P^{-T}\ones\end{bmatrix} \notag \\
            &= \begin{bmatrix} \alpha_1 b_1 & \cdots & \alpha_Y b_Y \end{bmatrix} \notag \\
            &= \begin{bmatrix} b_1 & \cdots & b_Y \end{bmatrix} \diag{[\alpha_1 \; \dots \; \alpha_Y]^T} \notag \\
            &= B \diag{[\alpha_1 \; \dots \; \alpha_Y]^T},
    \label{eq:relation_bar_B_and_B}
\end{align}
which is the observation matrix $B$ with scaled columns.

From \eref{eq:relation_bar_B_and_B}, it is clear that each column of $\bar{B}$ is colinear
with each corresponding column of $B$. Hence, we seek a diagonal matrix that properly
normalizes $\bar{B}$:
\begin{equation}
    B = \bar{B} \diag{d},
    \label{eq:postulated_relation_B_and_bar_B}
\end{equation}
where $d \in \Rb^Y$ is the vector of how much each column should be scaled. By
multiplying \eref{eq:postulated_relation_B_and_bar_B} from the right by $\ones$ and
employing the sum-to-one property of $B$, we obtain that the following should hold
\begin{align}
    B \ones = \bar{B} \diag{d} \ones &\implies \notag \\
    \ones = \bar{B} d.
\end{align}
Note that a solution to this equation exists by \eref{eq:relation_bar_B_and_B} and
that the $\alpha_y$:s are non-zero -- each element of $d$ is simply the inverse of each
$\alpha_y$.  Now, since $B$ is full column rank and the $\alpha_y$:s are non-zero,
relation \eref{eq:relation_bar_B_and_B} implies that $\bar{B}$ is also full column rank.
Hence, the unique vector of normalization factors $d$ is
\begin{equation}
    d = \bar{B}^\dagger \ones.
\end{equation}

\section{Proof of \lref{lem:dim_of_nullspace}}

Recall the following rank-result for Kronecker products\cite{horn_topics_1991}:
\begin{equation}
    \rank{A \otimes B} = \rank{A} \rank{B},
\end{equation}
which implies that 
\begin{equation}
    \rank{\pi_{k-1}^T \otimes [\pi_k \ones^T - I]} = 1 \times \rank{\pi_k \ones^T - I}.
\end{equation}
The last factor $\rank{\pi_k \ones^T - I}$ is equal to $X - 1$, since it is a rank-1
perturbation to the identity matrix.

\balance

\section{Proof of Remark~\ref{rem:obtaining_observations}}

To see that a unique observation can be reconstructed at each time $k$, suppose that
$\pi_k = T(\pi_{k-1}, y_k; P, B) = T(\pi_{k-1}, \tilde{y}_k; P, B)$ and observe that:
\begin{align}
    T(\pi, y; P, B) &= T(\pi, \tilde{y}; P, B) \notag \implies \\
\frac{\diag{b_{y}} P^T \pi}{\ones^T \diag{b_{y}} P^T \pi}
&= \frac{\diag{b_{\tilde{y}}} P^T \pi}{\ones^T \diag{b_{\tilde{y}}} P^T \pi}
\implies \notag \\
\diag{b_{y}} P^T \pi &= \alpha \diag{b_{\tilde{y}}} P^T \pi,
\label{eq:proof_obtain_obs_1}
\end{align}
where $\pi \in \Delta$ is a posterior and $\alpha \in \Rb_{> 0}$ a positive scalar.
Continuing, we have that equation \eref{eq:proof_obtain_obs_1} implies
\begin{align}
    \left( \diag{b_{y}} - \alpha \diag{b_{\tilde{y}}} \right) P^T \pi &= 0 \implies \notag \\
    \diag{b_{y} - \alpha b_{\tilde{y}}} P^T \pi &= 0 \implies \notag \\
    \left[ b_{y} - \alpha b_{\tilde{y}} \right]_i \left[P^T \pi\right]_i = 0,
\end{align}
for $i = 1, \dots, X$. Under \asref{as:P_B_positive}, we have that $\left[P^T \pi\right]_i
> 0$, so that we must have $\left[ b_{y} - \alpha b_{\tilde{y}} \right]_i = 0$, for $i
= 1, \dots, X$. Or, equivalently, 
\begin{equation}
    b_{y} = \alpha b_{\tilde{y}}.
\end{equation}
This yields, under \asref{as:P_B_full_column}, that $\alpha = 1$ and $y = \tilde{y}$.

\section{Proof of \thrmref{thrm:reconstruct_belief_from_action}}

In essence, \cite[Theorem 1]{mattila_estimating_2019} amounts to writing down the
Karush-Kuhn-Tucker conditions for \eref{eq:myopic_adversary} and considering the posterior
as an unknown variable.  It follows directly from this result that the adversary could
have held a belief $\pi_k$ when making the decision $u_k$ if and only if
\begin{equation}
    \pi_k \in \left\{ \pi \in \Delta : \sum_{i=1}^X [\pi]_i \nabla_u c(i, u_k) = 0 \right\}.
    \label{eq:adversary_belief_set}
\end{equation}
The set in \eref{eq:adversary_belief_set} can be rewritten on matrix-vector form as
\begin{align}
    \left\{ \pi \in \Rb_{\geq 0}^X:
        \begin{bmatrix}
            \nabla_\fact c(1, \fact_k) & \dots & \nabla_\fact c(X, \fact_k) \\
            1 & \dots & 1
        \end{bmatrix} \pi =
        \begin{bmatrix}
            0 \\
            \vdots \\
            0 \\
            1
        \end{bmatrix}
    \right\},
    \label{eq:adversary_belief_set_matrix_form}
\end{align}
which is non-empty by the fact that the adversary made a decision. Since, by
\asref{as:rank_of_cost_function}, the matrix
\begin{equation}
    F(\fact_k) =
    \begin{bmatrix}
        \nabla_\fact c(1, \fact_k) & \dots & \nabla_\fact c(X, \fact_k) \\
        1 & \dots & 1
    \end{bmatrix}
\end{equation}
has full column rank, the set \eref{eq:adversary_belief_set} -- or, equivalently
\eref{eq:adversary_belief_set_matrix_form} -- is singleton and the sole posterior in it is
\begin{equation}
    \pi_k = F(u_k)^\dagger  
        \begin{bmatrix}
            0 \\
            \vdots \\
            0 \\
            1
        \end{bmatrix},
\end{equation}
where $^\dagger$ denotes pseudo-inverse.

\end{document}

%% file: pre_rob.tex
\usepackage{graphicx}

\usepackage[cmex10]{amsmath}
\usepackage{amsfonts}
\usepackage{amssymb}
\usepackage{accents}
\usepackage{url}

\usepackage{empheq}
\usepackage[normalem]{ulem}

\usepackage{enumerate}

\usepackage{tikz}
\usepackage{pgfplots}
\usetikzlibrary{patterns}

\DeclareMathOperator{\diag}{diag}

\newcommand{\ind}[1]{\text{I}\{#1\}}

\usepackage{mathtools}
\DeclarePairedDelimiter{\PrfencesHard}{[}{]}

\DeclarePairedDelimiter{\PrfencesSoft}{(}{)}
\renewcommand{\Pr}{\operatorname{Pr}\PrfencesHard}

\renewcommand{\diag}{\operatorname{diag}\PrfencesSoft}
\renewcommand{\vec}{\operatorname{vec}\PrfencesSoft}
\newcommand{\rank}{\operatorname{rank}\PrfencesSoft}

\newcommand{\fref}[1]{Fig.~\ref{#1}}

\newcommand{\eref}[1]{(\ref{#1})}
\newcommand{\sref}[1]{Section~\ref{#1}}
\newcommand{\algoref}[1]{Algorithm~\ref{#1}}
\newcommand{\thrmref}[1]{Theorem~\ref{#1}}

\newcommand{\pref}[1]{Problem~\ref{#1}}

\newcommand{\lref}[1]{Lemma~\ref{#1}}

\ifIFAC
\usepackage{enumitem} 
\else

\makeatletter%
\@ifclassloaded{IEEEconf}{
\makeatother%

}{}

\makeatletter%
\@ifclassloaded{ieeeconf}{
\makeatother%

}{}

\usepackage{amsthm}
\fi

\theoremstyle{plain} 
\newtheorem{theorem}{Theorem}
\newtheorem{lemma}{Lemma}

\theoremstyle{definition} 

\newtheorem{problem}{Problem}
\newtheorem{assumption}{Assumption}

\theoremstyle{remark} 
\newtheorem{remark}{Remark}

\newcommand{\Rb}{{\mathbb R}}

\usepackage{dsfont}
\newcommand{\ones}{\mathds{1}}

\newcommand{\eqdef}{\overset{\text{def.}}{=}}  

\usepackage{subcaption}